\documentclass[preprint,12pt,compress]{elsarticle}
\usepackage{lineno}
\usepackage{amsmath,amsfonts}
\allowdisplaybreaks
\usepackage{algorithmic}
\usepackage{array}
\usepackage{textcomp}
\usepackage{stfloats}
\usepackage{url}
\usepackage{verbatim}
\usepackage{graphicx}
\usepackage{makecell}
\usepackage{epstopdf}
\usepackage{dcolumn}
\usepackage{bm}
\usepackage{natbib}
\setcitestyle{numbers,square}
\usepackage{braket}
\usepackage{color}
\usepackage{algorithm}
\usepackage{balance}
\usepackage{booktabs}
\usepackage{amssymb}
\usepackage{amsthm}
\usepackage{caption}
\usepackage{tikz}
\usepackage{adjustbox}
\usepackage{titlesec}
\usepackage{quantikz}
\usepackage{hyperref}

\theoremstyle{plain}
\newtheorem{theorem}{Theorem}

\newtheorem{lemma}{Lemma}
\newtheorem{proposition}{Proposition}

\theoremstyle{definition}
\newtheorem{definition}{Definition}

\newtheorem{remark}{Remark}

\journal{}

\begin{document}

\begin{frontmatter}



\title{{Low-Depth} Distributed Quantum Algorithms\\ for Unordered Database Search} 


\author[mymainaddress,mythirdaddress]{Huaijing Huang}
\author[mymainaddress,mythirdaddress]{Daowen Qiu\corref{mycorrespondingauthor}}
\ead{issqdw@mail.sysu.edu.cn}
\cortext[mycorrespondingauthor]{Corresponding author}
\author[mymainaddress,mythirdaddress]{Ximing Hua}
\author[mymainaddress,mythirdaddress]{Xinyu Chen}

\affiliation[mymainaddress]{organization={School of Computer Science and Engineering},
            addressline={Sun Yat-sen University}, 
            city={Guangzhou},
            postcode={510006}, 
            country={China}}
\affiliation[mythirdaddress]{organization={The Guangdong Key Laboratory of Information Security Technology},
            addressline={Sun Yat-sen University}, 
            city={Guangzhou},
            postcode={510006}, 
            country={China}}
\begin{abstract}
Grover's algorithm accelerates unstructured database search quadratically compared to  classical algorithms. In the NISQ era, distributed quantum computing can decrease circuit depth and reduce noise.  In this paper, an algorithm for constructing  query operators for subfunctions is proposed.  By dividing the target string of the search problem into several substrings and integrating the query operator of each subfunction, {we design} a low-depth distributed exact quantum search algorithm. The contributions of this paper are as follows: 
(1) The proposed distributed algorithm has a lower circuit depth and can mitigate error accumulation compared   to {other} distributed quantum search algorithms; (2) The target can be accurately located by the proposed distributed algorithm; (3) Experiments conducted with the quantum software MindQuantum confirm the effectiveness and feasibility of the proposed distributed algorithm. Moreover, the introduction of  noise {into} the circuit during these experiments indicates that the algorithm possesses an inherent capacity for noise resistance.
\end{abstract}


\end{frontmatter}



\section{Introduction}
Quantum computing has become a highly sought-after field of research globally.  Compared with classical computing, quantum computing is  superior.  {However,  in the noisy intermediate-scale quantum (NISQ) era\cite{Preskill2018}, quantum computers suffer from insufficient qubits and difficulty in isolating qubits from noise.} To solve these problems, distributed quantum computing is a better method that can simulate large-scale circuits  and reduce noise\cite{Li2023}.  Distributed quantum algorithms, which serve as the core of distributed quantum computing, have attracted considerable interest among researchers.  Distributed quantum algorithms have been designed from various perspectives,  including Boolean function decomposition,  problem decomposition, and circuit decomposition\cite{Li2017,Avron2021,TAN_quantum_2022,Qiu2022,Xiao2023,Li2024,Le2019}.

In recent years, distributed quantum algorithms have been developed rapidly. In 2017, Li et al. considered a distributed phase estimation algorithm that requires only classical communication\cite{Li2017}. In 2020, Izumi et al. demonstrated a  distributed quantum algorithm for solving the triangle finding problem in the CONGEST model\cite{Le2019}. In the same year, Avron et al. studied the quantum advantages in distributed quantum computing\cite{Avron2021}. In 2022, Qiu et al. proposed a distributed Grover's search algorithm\cite{Qiu2022}, and  Tan et al. designed a distributed Simon algorithm\cite{TAN_quantum_2022}.  In 2023, Xiao et al. presented a distributed Shor algorithm\cite{Xiao2023}. Andres-Martinez et al. explored the implementation of distributed quantum circuits in modular quantum computing network architectures\cite{P2024}. Notably, Qiu et al. proposed an error correction method for distributed quantum computing in 2025 \cite{qiu2025universal}. In addition, two distributed quantum algorithms have been proposed, including distributed generalized Deutsch--Jozsa\cite{li2024distributed} and distributed amplitude amplification  algorithms \cite{HQ2025}.


 Grover's algorithm\cite{Grover1996} was initially proposed to solve the problem of searching in an unsorted database, but it does not always search the target with a success probability of 1. Later, Long improved upon  Grover's algorithm and proposed an exact search algorithm, which can search the target with {a} success probability of 1\cite{Long2001}.   Another extension of Grover's algorithm is amplitude amplification\cite{Brassard2002},  which can be applied as a subroutine to other quantum algorithms\cite{Ambainis2004, Biamonte2017}.

{A} quantum partial search algorithm only obtains partial information {about} the target. It was first proposed by Grover and Radhakrishnan\cite{GRK_2005} in 2005.  They divided the entire database  into several equal blocks and iteratively {located} the target block (containing the target item) using global and local search operators. The global search operator is {applied to}  the entire database and is actually the search operator in primitive Grover's algorithm. The local search operator is {applied} in parallel {to} each block. 
Korepin\cite{Korepin_2005} then optimized the number of queries for the quantum partial search algorithm.  Subsequently, other researchers have improved and streamlined certain partial search algorithms from other perspectives\cite{Choi_2007,Zhang_2018}.  In 2020, Zhang et al.\cite{Zhang_2020} proposed an optimized Grover’s search algorithm in terms of depth.

Considering that deep quantum circuits are difficult to execute on NISQ devices because of limited coherence time, reducing circuit depth is crucial for practical quantum search algorithms.  Motivated by this, we employ a distributed computing strategy to reduce the circuit depth required by Grover's search algorithm. Specifically, we design an exact distributed quantum algorithm for solving the unordered database search problem while maintaining the exactness of the search results.

The paper is structured in the following manner. In Section \ref{s2}, two quantum search algorithms will be reviewed briefly:   the modified Grover’s algorithm by Long and {the} GRK quantum partial search algorithm.   An exact quantum partial search algorithm {is} proposed in Section \ref{s3}. In Section \ref{s4},  an algorithm for constructing a query operator from Boolean functions to their subfunctions is presented first. {Then, we} design a  low-depth distributed exact search algorithm  to extract the target.   We also prove the correctness of the algorithm and  compare our algorithm to existing exact quantum search algorithms. The proposed algorithm is subjected to both noisy and noiseless experiments in Section \ref{s5}.  Conclusions and  prospects are presented in Section \ref{s6}.

\section{Preliminaries}\label{s2}
In this section,  we  begin with a brief description of  {the} modified Grover’s algorithm by Long\cite{Long2001}, followed by {the} GRK quantum partial search algorithm\cite{Korepin_2005}.

 \subsection{ Modified Grover’s algorithm by Long}
 Due to the fact that Grover's  algorithm\cite{Grover1996} does not always find the target with a probability of 1, Long improved upon this algorithm by proposing an exact search algorithm\cite{Long2001}.
 Below, we provide a detailed introduction to the modified Grover’s algorithm by Long.

Assume an unsorted database with $N = 2^n$ items, out of which {an} item is marked. Our goal is to find the marked item.   The corresponding Boolean function for this problem can be constructed as  $f:\{0,1\}^n \rightarrow \{0,1\}$, with $\lvert\{ t \in \{0,1\}^n|f(t) = 1 \}\rvert=1$. The corresponding question {is to}  {find} the target $t\in \{0,1\}^n$ that makes $f(t) = 1$. In quantum computing, for a given function $f$, there is an oracle $U^{\omega}_{f}$ that {queries}
$|{x}\rangle \to e^{\imath\omega\cdot f(x)}|x\rangle$ for any $ x\in\{0,1\}^n$.
 
We denote 
$$|\hat{x}\rangle:= \frac  {1}{\sqrt{2^n-1}}\sum_{x\neq t}|x\rangle.$$
The initial state of  {the} algorithm is\begin{equation}\label{in1}
|\varphi_n\rangle=H^{\otimes n}|0^{n}\rangle=\frac  {\sqrt{2^n-1}} {\sqrt {2^n}}|\hat{x}\rangle+\frac  {1} {\sqrt {2^n}}|t\rangle.
\end{equation}
 Let $ \lambda=\arcsin{{\frac {1}{\sqrt {2^n}}}}$.   Eq. (\ref{in1}) then becomes \begin{equation}\label{in2}
 |\varphi_n\rangle=\cos{\lambda}|\hat{x}\rangle+\sin{\lambda}|t\rangle.
 \end{equation}
 Define the iterative operator $L$ for the modified Grover’s algorithm as  \begin{align}\label{in3}
 L &=-H^{\otimes n}\left((e^{\imath\omega
}-1)|0^{n}\rangle\langle0^{n}|+I_{n}\right)H^{\otimes n}U_{f}^{\omega}\\
  &=\left((1-e^{\imath\omega
})|\varphi_n\rangle\langle\varphi_n|-I_{n}\right)U_{f}^{\omega},\end{align}   where \begin{align*}
U_{f}^{\omega}|x\rangle &=
\begin{cases}
e^{\imath\omega
\cdot f(x)}|x\rangle, & f(x)=1,\\
|x\rangle, & f(x)=0,
\end{cases} 
 \end{align*}
 $I_{n}$ is the $2^n$-dimensional identity operator. The phase inversion angle needs to satisfy $$\omega=2\arcsin\left(\sin\left(\frac{\pi}{4\lfloor(\pi/2-\lambda)/(2\lambda)\rfloor+6}\right) / \sin  \lambda\right).$$ We just need to apply the operator $L$ to the initial state {in} Eq. (\ref{in2}) {for} $\lfloor(\pi/2-\lambda)/(2\lambda)\rfloor+1$ {iterations}, and {then} we can obtain the target with certainty through measurement. In fact, the number of iterations of {the} $L$ operator in the improved Grover’s algorithm by Long   can be any integer equal to or greater than $\lfloor(\pi/2-\lambda)/(2\lambda)\rfloor$.  {The} detailed {algorithmic procedure} is described in Algorithm \ref{algo1}.
 \begin{algorithm}
\caption{ The modified Grover’s algorithm by Long}
\label{algo1}
\begin{algorithmic}
\STATE \textbf{Input}:  Boolean function $f:\{0,1\}^n \rightarrow\{0,1\}$ with $|\{x\in\{0,1\}^n|f(x)=1\}|= 1$.
\STATE \textbf{Output}: The string $t$ satisfies $f(t)=1$ with probability  1.
\STATE \textbf{Procedure:}
\STATE Step 1:  Execute the iterative operator $L$  on {the} initial state $|\varphi_n\rangle$ for  $\lfloor(\pi/2-\lambda)/(2\lambda)\rfloor+1$  {iterations}.
\STATE Step 2: Measure each qubit in the $\{|0\rangle,|1\rangle\}$ basis.
\end{algorithmic}
\end{algorithm}

\subsection{GRK quantum partial search algorithm}
 
 Assume an unordered database has $N=2^n$ items. There is only one marked item that is the target. The goal is to find the first $q$ bits of the target, where $q\in N^{+}$. We express this problem  in the form of a definition. 
 
 \begin{definition}\label{defi2}
Given a Boolean function $f:\{0,1\}^n \rightarrow \{0,1\}$ with $\lvert\{ t \in \{0,1\}^n|f(t) = 1 \}\rvert\ge1$,   partial search problems involve finding the substring $t_{0}t_{1}\cdots t_{q-1}$ such that  $t_{0}t_{1}\cdots t_{q-1}t'\in\{ t \in \{0,1\}^n|f(t) = 1 \}$, where $t'\in \{0,1\}^{n-q}$, $1\le q \le n$.
\end{definition}
 If $q =n$, then we {revert  to} the problem that  Grover’s  algorithm aims to solve. A quantum partial search algorithm solves this problem.  {Such an algorithm} was proposed by Grover and Radhakrishnan\cite{GRK_2005} in 2005.  It was subsequently optimized by Korepin\cite{Korepin_2005}, {and the resulting method is referred to as the} GRK quantum partial search algorithm. We consider the case of a single target, that is, $\lvert\{ t \in \{0,1\}^n|f(t) = 1 \}\rvert=1$.
  
  First,   the items {in} $\{0,1\}^n$   {are} divided  into $2^q$ blocks based on whether the first $q$ bits are the same, as follows: $$ \overbrace{0 \cdots000}^{q}y_1,~ \overbrace{0 \cdots001}^{q}y_2, ~\cdots, ~ \overbrace{1 \cdots111}^{q}y_{2^{q}}, y_m \in \{0,1\}^{n-q}, m=1,\cdots,2^{q}.$$
 The target block is the one containing the target item.   {Therefore,} we only need to find the  block where  $t_{0}t_{1}\cdots t_{q-1}$ resides.  Let the target substring we want to find be denoted by \( t^* \), where \( t^* =t_{0}\cdots t_{q-1} \). 
 
 Next, we  define some operators involved in the GRK quantum partial search algorithm.
 
{The global Grover operator} $G$  is {defined} as follows:  \begin{equation}G=U_{n}U_{f},\end{equation}  where \begin{equation}U_{n}=2|\varphi_n\rangle\langle\varphi_n|-I_{n},\end{equation}
 \begin{equation}U_{f}|x\rangle=(-1)^{f(x)}|x\rangle, \forall x\in\{0,1\}^{n},\end{equation}$I_{n}$ is the $2^n$-dimensional identity operator.
 
{The local Grover operator} $G_1$  is defined as follows: \begin{equation} \label{def G2}
G_1=(I_q\otimes U_{{n-q}})U_{f},
\end{equation}
where \begin{equation}U_{n-q}=2|\varphi_{n-q}\rangle\langle\varphi_{n-q}|-I_{n-q}, \end{equation}
$|\varphi_{n-q}\rangle$ is a uniform superposition state in one block,  \begin{equation}|\varphi_{n-q}\rangle={\frac{1}{\sqrt{2^{n-q}}}}\sum_{x=0}^{2^{n-q}-1}|x\rangle.\end{equation} 

The detailed process of  {the} GRK quantum partial search algorithm is given {in}  Algorithm \ref{algo2}.

\begin{algorithm}
\caption{ GRK quantum partial search algorithm}
\label{algo2}
\begin{algorithmic}
\STATE \textbf{Input}:  Boolean function $f:\{0,1\}^n \rightarrow\{0,1\}$ with $|\{x\in\{0,1\}^n|f(x)=1\}|= 1$.
\STATE \textbf{Output}: The string $t_{0}\cdots t_{q-1}$ satisfies $f(t_{0}\cdots t_{q-1}\cdots t_{n-1})=1$ with high probability.
\STATE \textbf{Procedure:}
\STATE Step 1:  Execute {the global} Grover operator $G$  on {the} initial state $|\varphi_n\rangle$ for $j_1$ {iterations} to obtain $|\psi_1\rangle$.
\STATE Step 2:  Execute the local Grover operator $G_1$  on {the} state $|\psi_1\rangle$ for $j_2$ {iterations} to obtain $|\psi_2\rangle$.
\STATE Step 3:  Execute the {global} Grover operator $G$    on {the} state $|\psi_2\rangle$  {once}  to obtain $|\psi_3\rangle$. 
\STATE Step 4:  {Measuring} the first $q$ qubits of  $|\psi_3\rangle$ yields the string $t_{0}\cdots t_{q-1}$ with high probability.
\end{algorithmic}
\end{algorithm}

 By applying $G^{j_1}$ and {then} ${G_1}^{ j_2}$ to the initial state $|\varphi_n\rangle$, we obtain
\begin{equation}\label{a16}
{G_1}^{ j_2}G^{j_1}|\varphi_n\rangle= |t^*\rangle\otimes\left(g_{t}|t'\rangle
+ b_{t} \sum_{x\neq t'} \frac{|x\rangle}{\sqrt{2^{n-q}-1}}\right)
+   \sum_{x^{''}\neq t^*}|x^{''}\rangle\otimes \frac{\sqrt{2^{n-q}}\cos \left(\left(2 j_1+1\right) \lambda\right)}{\sqrt{2^{n}-1}}|\varphi_{{n-q}}\rangle,
\end{equation}
where $$g_{t}=\sin \left(\left(2  j_1+1\right) \lambda\right)\cos \left(2j_2\lambda'\right)+\frac{\sqrt{2^{n-q}-1}}{\sqrt{2^{n}-1}}\cos \left(\left(2  j_1+1\right) \lambda\right) \sin \left(2j_2\lambda'\right),$$
$$b_{t}=-\sin \left(\left(2  j_1+1\right) \lambda\right)\sin \left(2j_2\lambda'\right)+\frac{\sqrt{2^{n-q}-1}}{\sqrt{2^{n}-1}}\cos \left(\left(2  j_1+1\right) \lambda\right) \cos \left(2j_2\lambda'\right),$$
$$\lambda'=\arcsin{\frac {1} {\sqrt{2^{n-q}}}}.$$

To make the amplitude of the non-target blocks zero, the parameters  $j_1$ and $j_2$ in Algorithm \ref{algo2} need to meet {a} certain condition, which we represent with a proposition.

\begin {proposition}\label{p1}
  If $j_1$ and $j_2$ satisfy the following equation
\begin{equation}\label{eqa cancellation different}
\frac {2^{n}} {\sqrt{2^{n}-1}}\left(\frac{1} {2^{q}}-\frac 1 2\right)\cos\left((2j_1+1)\lambda\right) =-g_{t}+b_{t}\sqrt{2^{n-q}-1},
\end{equation} 
then  the amplitude of the non-target block is zero.
\end {proposition}

 The proof of Proposition \ref{p1} can be found in \ref{app1}. Korepin provided the optimized number of queries, which we denote using a proposition.
 
\begin {proposition}\label{proposition2}
If  $q$ is a finite number, $2^{n}\rightarrow\infty$, {and $j_1$ and $j_2$ are given by}\begin{equation}\label{628} 
 j_1 = \frac \pi 4 {\sqrt{2^{n}}}-\alpha_{q}{\sqrt {2^{n-q}}},\quad\quad   j_2 =  \beta_{q}{\sqrt {2^{n-q}}},\end{equation}
 where \begin{equation*}
\alpha_{q}=\frac{\sqrt {2^q}}{2}\arctan\left(\frac{\sqrt{3\cdot{2^q}-4}}{{2^q}-2}\right),\quad\quad \beta_{q}=\arcsin\sqrt{\frac{{2^q}}{4({2^q}-1)}},
\end{equation*}    then {the} GRK quantum partial search algorithm  can return the  correct answer with probability 1, 
using the minimum number of queries \begin{equation*}
  \frac \pi 4 {\sqrt{2^{n}}}+ \left({\arcsin\sqrt{\frac{{2^q}}{4({2^q}-1)}}-\frac{\sqrt {2^q}}{2}\arctan\left(\frac{\sqrt{3\cdot{2^q}-4}}{{2^q}-2}\right)}\right){\sqrt {2^{n-q}}}+1.
\end{equation*}
\end {proposition}
The proof of Proposition \ref{proposition2} can be found in \ref{x2}.

\section{ Exact quantum partial search algorithm}\label{s3}
Although Ref.\cite{Choi_2007} proposed a quantum partial search algorithm with {a} success probability {of} 1,  a rigorous mathematical proof has not been provided. {The authors}   did not provide specific expressions for the query numbers $j_1$ and $j_2$ in {the} algorithm.  They demonstrated through numerical experiments that when $N\le10^6$,  the algorithm's success rate reaches 1 if   appropriate $j_1$ and $j_2$ are selected. We provide a detailed algorithm in which $j_1$ and $j_2$ are specified. We make modifications to the iterative operator in the final step of the GRK partial search algorithm, and derive an explicit analytical solution for the parameters of the final iterative operator. Moreover, we provide a theorem that explicitly states the conditions that the algorithm must satisfy to achieve accuracy. In the following, we provide a detailed introduction to this algorithm.

Assume an unordered database has $N=2^n$ items and there is only one marked item as the target. We  improve the GRK quantum partial search algorithm to accurately obtain the desired partial information. We introduce the {generalized global Grover operator} $G_g$  in the final iteration of the GRK quantum partial search algorithm. 
{The  generalized global Grover operator $G_g$ is defined} as follows:
\begin{equation}\label{generalizG}
G_g=\left((1-e^{\imath\theta })|\varphi_{{n}}\rangle\langle\varphi_{{n}}|-I_{{n}}\right)U_{f}^{\phi}, \quad \phi, \theta\in \mathbb{R},
\end{equation}
where   \begin{align*}
U_{f}^{\phi}|x\rangle &=
\begin{cases}
e^{\imath\phi \cdot f(x)}|x\rangle, & f(x)=1,\\
|x\rangle, & f(x)=0.
\end{cases} 
\end{align*}

It should be noted that the {numbers} of iterations in {Steps} 1 and 2 {are} determined by finding the closest positive integers based on the query counts of Proposition \ref{proposition2}.  Let $$ j_1^{'} =\lfloor \frac \pi 4 {\sqrt{2^{n}}}-\alpha_{q}{\sqrt {2^{n-q}}}\rceil,  j_2^{'} = \lfloor \beta_{q}{\sqrt {2^{n-q}}}\rceil.$$ For the detailed algorithm, see Algorithm \ref{algo3}.

\begin{algorithm}
\caption{ The exact quantum partial search algorithm}
\label{algo3}
\begin{algorithmic}
\STATE \textbf{Input}:  Boolean function $f:\{0,1\}^n \rightarrow\{0,1\}$ with $|\{x\in\{0,1\}^n|f(x)=1\}|= 1$.
\STATE \textbf{Output}: The string $t_{0}\cdots t_{q-1}$ satisfies $f(t_{0}\cdots t_{q-1}\cdots t_{n-1})=1$ with a probability of 1.
\STATE \textbf{Procedure:}
\STATE Step 1:  Execute {the global} Grover operator $G$  on {the}  initial state $|\varphi_n\rangle$ for $j_1^{'}$ {iterations} to obtain $|\psi_1\rangle$.
\STATE Step 2:  Execute the local Grover operator $G_1$  on {the} state $|\psi_1\rangle$ for $j_2^{'}$ {iterations} to obtain $|\psi_2\rangle$.
\STATE Step 3:  Execute the {generalized global} Grover operator $G_g$   on {the}  state $|\psi_2\rangle$ {once}  to obtain $|\psi_3'\rangle$. 
\STATE Step 4: {Measuring} the first $q$ qubits of $|\psi_3'\rangle$ yields the string $t_{0}\cdots t_{q-1}$.
\end{algorithmic}
\end{algorithm}

Algorithm \ref{algo3} can output the desired result with {a} success probability {of} 1 {when}   appropriate $\phi$  and $\theta$ {are selected}. Next, in order to obtain $t_{0}\cdots t_{q-1}$ with {a}  success probability {of} 1, we need to adjust the parameters $\phi$  and $\theta$.

As in Eq. (\ref{a16}), we have
\begin{equation}
{G_1}^{ j_2^{'}}G^{j_1^{'}}|\varphi_n\rangle= a_{t^{'}}|t^*t'\rangle
+ a_{n_{t^{'}}} \sum_{x\neq t'} \frac{|t^*x\rangle}{\sqrt{2^{n-q}-1}}
+   \sum_{x^{''}\neq t^*}|x^{''}\rangle\otimes \frac{\sqrt{2^{n-q}}\cos \left(\left(2 j_1^{'}+1\right) \lambda\right)}{\sqrt{2^{n}-1}}|\varphi_{{n-q}}\rangle,
\end{equation}
where $$a_{t^{'}}=\sin \left(\left(2  j_1^{'}+1\right) \lambda\right)\cos \left(2j_2^{'}\lambda'\right)+\frac{\sqrt{2^{n-q}-1}}{\sqrt{2^{n}-1}}\cos \left(\left(2  j_1^{'}+1\right) \lambda\right) \sin \left(2j_2^{'}\lambda'\right),$$
$$a_{n_{t^{'}}}=-\sin \left(\left(2  j_1^{'}+1\right) \lambda\right)\sin \left(2j_2^{'}\lambda'\right)+\frac{\sqrt{2^{n-q}-1}}{\sqrt{2^{n}-1}}\cos \left(\left(2  j_1^{'}+1\right) \lambda\right) \cos \left(2j_2^{'}\lambda'\right).$$

In order to make the {amplitudes} of {the} non-target blocks equal to  zero, we first propose the following lemma. 

We denote 
$$E' :=\sqrt{2^{n-q}-1}\cdot a_{n_{t^{'}}}+\frac{\left(2^{n}-2^{n-q}\right)\cos \left(\left(2  j_1^{'}+1\right) \lambda\right)}{\sqrt{2^{n}-1}},$$
$$F' :=\frac{\cos \left(\left(2  j_1^{'}+1\right) \lambda\right)}{\sqrt{2^{n}-1}}.$$

\begin {lemma}\label{l3}
If $E'$,  $F'$ and $a_{t^{'}}$ satisfy the  inequality 
\begin{equation}\label{gzong1}
\left(E'-{2^{n-1}F'}\right)^2\le a_{t^{'}}^2,
\end{equation}
let  {$\phi$ and $\theta$ satisfy} \begin{equation}\label{xzong14} \left\{
\begin{aligned}
&\cos\theta=\frac{2^{2n-1}F'^2-2^{n}E'F'+E'^2-a_{t^\prime}^2}{E'^2-2^{n}E'F'-a_{t^\prime}^2}, \\
&\cos\phi=\frac{2^{n-1}F'-E'}{a_{t^\prime}},\\
&a_{t^{'}}\sin\phi\sin\theta>0.
\end{aligned}
\right.
\end{equation}
  {Then,}  $\phi$ and $\theta$   necessarily satisfy {the}  {equations} \begin{equation}\label{sgzong14} \left\{
\begin{aligned}
\cos \left(\phi+\frac{\theta }{2}\right) & =  \frac{-E'\cos\frac{\theta }{2}} {a_{t^{\prime}}}, \\
\sin\left(\phi+\frac{\theta}{2}\right)& =\frac{\left(\cos\theta-1 \right)E'+2^{n}F'}{2a_{t^{'}}\sin\frac{\theta}{2}}.
\end{aligned}
\right.
\end{equation}
\end{lemma}

The proof of Lemma \ref{l3} is available in  \ref{x3}. In the following, we  elucidate how to precisely obtain the first $q$ bits of the target through {a} theorem. 

\begin{theorem}\label{gt3}
 If  $\phi$ and $\theta$    satisfy {the following} {equations} \begin{equation*} \left\{
\begin{aligned}
\cos \left(\phi+\frac{\theta }{2}\right) & =  \frac{-E'\cos\frac{\theta }{2}} {a_{t^{\prime}}}, \\
\sin\left(\phi+\frac{\theta}{2}\right)& =\frac{\left(\cos\theta-1 \right)E'+2^{n}F'}{2a_{t^{'}}\sin\frac{\theta}{2}},
\end{aligned}
\right.
\end{equation*}  then we can obtain the first $q$ bits $t_{0}\cdots t_{q-1}$ that satisfy {the} Boolean function $$f(t_{0}\cdots t_{q-1}\cdots t_{n-1})=1$$ with probability 1.
\end{theorem}
The proof of Theorem \ref{gt3} is available in  \ref{x4}.

\section{Iterative exact distributed search algorithm}\label{s4}
 To  reduce the circuit depth, we  design a simpler algorithm with more straightforward query operators. In this section, an algorithm for constructing a query operator from a Boolean function to its subfunctions is presented first.
We then give an iterative exact distributed   search algorithm. Subsequently,  we   prove the correctness of the algorithm.  Finally, we juxtapose it with existing precise  search algorithms for a comparative evaluation.
\subsection{ Algorithm design }
Assume an unsorted database has $N=2^{n}$ items, only one of which is the target. We divide the target  string into several segments and use an iterative approach to find the entire target string. The initial segments are processed  using the exact   partial search algorithm, while the final part is searched using the modified Grover’s algorithm by Long. We describe the method in detail next.

\begin{definition}\label{defi}
Given a Boolean function $f:\{0,1\}^n \rightarrow \{0,1\}$, define the corresponding subfunctions $f_{i}$ as  follows: $f_{i}:\{0,1\}^{n-k} \rightarrow \{0,1\}$, $i\in\{0,1\}^{k}$, $k\in N^+$, where  $f_{i}$ satisfies $f_{i}(x)=f\left( xi\right)$,  $\forall x \in \{0,1\}^{n-k}$. 
\end{definition}
  The subfunction that contains the target is called the target subfunction, denoted by $f_{ \widetilde{i}}$.   The target  state $|t\rangle$ of {the target} subfunction is denoted as $|x_{ \widetilde{i}}\rangle\otimes|y_t\rangle (p\le n-k, p\in N^{+})$, where $x_{ \widetilde{i}}=x_{0}\cdots x_{p-1}\in\{0,1\}^{p}$,  and  $y_t\in\{0,1\}^{n-p-k}$. Initially, we  determine {the} target substring $x_{ \widetilde{i}}$. 
 {The}  oracle $U_{f_{i}}$($i\in\{0,1\}^{k}$) can query  $f_{i}(x)=f\left(xi\right)$ for all $x \in \{0,1\}^{n-k}$, where $xi$  represents   {the concatenation of} the bit strings  $x$ and   $i$. 
  
In quantum query algorithms, there are two common types of {oracles}: one with  {an} auxiliary bit,  represented as $U_{f}|x\rangle|0\rangle\rightarrow|x\rangle|f(x)\rangle$, and the other without {an} auxiliary bit, represented as $U_{f}|x\rangle\rightarrow(-1)^{f(x)}|x\rangle$. When the auxiliary bit of $U_{f}$ is in {the} state $|-\rangle$, the oracle $U_{f}$ with  {an} auxiliary bit is equivalent in effect to $U_{f}$ without {an} auxiliary bit, that is, $U_{f}|x\rangle|-\rangle\rightarrow(-1)^{f(x)}|x\rangle|-\rangle$.  

Next, we present an algorithm for obtaining the output of $U_{f_{i}}^{\alpha}$ from $U_{f}$  with {an} auxiliary bit. {$U_{f_{i}}^{\alpha}$ is defined as $U_{f_{i}}^{\alpha}|x\rangle|0\rangle\rightarrow e^{\imath\alpha \cdot {f_{i}(x)}}|x\rangle|{f_{i}(x)}\rangle$, where  $x\in\{0,1\}^{n-k}$}, $i\in\{0,1\}^{k}$, $\alpha\in \mathbb{R}$. 
\begin{algorithm}
\caption{{ Algorithm for constructing $U_{f_{i}}^{\alpha}$ from $U_{f}$ }}
\label{algo22}
\begin{algorithmic}
\STATE \textbf{Input}:  $U_{f}$.
\STATE \textbf{Output}: {$e^{\imath\alpha \cdot {f_{i}(x)}}|x\rangle|{f_{i}(x)}\rangle$}.
\STATE \textbf{Procedure:}
\STATE Step 1:  Prepare the initial state {$|x\rangle|i\rangle|0\rangle$. }
\STATE Step 2:  Execute the  operator $U_{f}$ to obtain  { $ {|x\rangle|i\rangle|f(x i)\rangle}$}. 
\STATE Step 3: Execute the  operator $I_{n-k}\otimes Swap$ to obtain {${|x\rangle|f(x i)\rangle|i\rangle}$}.
\STATE Step 4: Apply  $X$ gates to the qubits corresponding to the $1$-bits of $i$ to  obtain {${|x\rangle|f(x i)\rangle|0\rangle}$.}
\STATE Step 5: Execute the operator $I_{n-k}\otimes  R$ to obtain {$e^{\imath\alpha \cdot f(x i)}{|x\rangle|f(x i)\rangle}$}. 
\end{algorithmic}
\end{algorithm}

  In Algorithm \ref{algo22}, \begin{equation*}R=
\begin{pmatrix} 1 & 0 \\ 0 & e^{\imath\alpha} \end{pmatrix},\end{equation*}  and Swap is the swap gate. {Since $f_{i}(x)=f\left( xi\right)$ by} the definition of the subfunction, Algorithm \ref{algo22} produces {$e^{\imath\alpha \cdot f_{i}(x)}{|x\rangle|f_{i}(x)\rangle}.$}
  
In the following distributed algorithm, we assume there exists such a $U_{f_{i}}^{\alpha}$ without an auxiliary bit.  First, we will define some operators.

The definition of {the} generalized global Grover  operator $G_4$ is as follows:
\begin{equation}\label{generalizG}
G_4=\left((1-e^{\imath\theta })|\varphi_{{n-k}}\rangle\langle\varphi_{{n-k}}|-I_{{n-k}}\right)U_{f_{i}}^{\phi}, \quad \phi, \theta\in \mathbb{R},
\end{equation}
where $|\varphi_{{n-k}}\rangle=H^{\otimes n-k}|0^{n-k}\rangle$ and $\forall i\in\{0,1\}^{k}$, \begin{align*}
U_{f_{i}}^{\phi}|x\rangle &=
\begin{cases}
e^{\imath\phi \cdot f_i(x)}|x\rangle, & f_i(x)=1,\\
|x\rangle, & f_i(x)=0.
\end{cases} 
\end{align*}
 Specifically, if  $\theta=\phi=\pi$, then \begin{equation}\label{generalizG2}G_2=\left(2|\varphi_{{n-k}}\rangle\langle\varphi_{{n-k}}|-I_{{n-k}}\right)U_{f_{i}}.\end{equation}

We set the iteration counts for {Steps} 3 and 4  as follows: 
 \begin{equation}\label{cqu1}
p_1=\lfloor\frac \pi 4\sqrt{ 2^{n-k}}- \gamma_p\sqrt {2^{n-p-k}}\rceil, \quad p_2=\lfloor\eta_p\sqrt{2^{n-p-k}}\rceil,  \end{equation}  
 where \begin{equation}\label{m2}
\gamma_p=\frac{\sqrt {2^{p}}}{2}\arctan\left(\frac{\sqrt{3\cdot2^{p}-4}}{2^{p}-2}\right),\quad\quad \eta_p=\arcsin\left(\frac{\sqrt {2^{p}}}{\sqrt{4(2^{p}-1)}}\right).
\end{equation}  

 The local Grover  operator $G_3$ is  constructed as follows{:} \begin{equation} \label{def G3}
G_3=\left(I_{{p}}\otimes U_{{{n-p-k}}}\right)U_{f_i},
\end{equation}
 where  $U_{{{n-p-k}}}=2|\varphi_{{n-p-k}}\rangle\langle\varphi_{{n-p-k}}|-I_{{n-p-k}}$, $|\varphi_{{n-p-k}}\rangle=H^{\otimes n-p-k}|0^{n-p-k}\rangle$ is a uniform superposition state.

The operator $L_i$ is defined as  $$L_i=\left((1-e^{\imath\omega
})|\varphi_{{n-p-k}}\rangle\langle\varphi_{{n-p-k}}|-I_{{n-p-k}}\right)U_{f_{i,x_i}}^{\omega}$$ for any $ i\in\{0,1\}^{k}$, where 
$$\omega=2\arcsin\left(\sin\left(\frac{\pi}{4\lfloor(\pi/2-\theta')/(2\theta')\rfloor+6}\right) / \sin  \theta'\right),$$
$$|\varphi_{{n-p-k}}\rangle=H^{\otimes n-p-k}|0^{n-p-k}\rangle,$$
\begin{align*}
U_{f_{i,x_i}}^{\omega}|y\rangle &=
\begin{cases}
e^{\imath\omega \cdot f_{i,x_i}(y)}|y\rangle, & f_{i,x_i}(y)=1,\\
|y\rangle, & f_{i,x_i}(y)=0.
\end{cases} 
\end{align*} 
Here, $\theta'=\arcsin{\frac {1} {\sqrt{2^{n-p-k}}}}$ and ${f_{i,x_i}}(y)= {f_{i}}(x_iy), \forall y\in\{0,1\}^{n-p-k}.$ $x_i$ is a specific $p$-bit string obtained through {Steps} 1 to 6.

\begin{algorithm}[htbp]
\caption{ Iterative exact distributed   search algorithm }
\label{algo21}
\begin{algorithmic}
\STATE \textbf{Input}:  Boolean function $f:\{0,1\}^n \rightarrow\{0,1\}$ with $|\{x\in\{0,1\}^n|f(x)=1\}|= 1$.
\STATE \textbf{Output}: The string $x= x_{ \widetilde{i}}y_t\widetilde{i}$ satisfies $f(x)=1$ with a probability of 1.
\STATE \textbf{Procedure:}
\STATE Step 1:  {The}  $2^k$ machines will execute the following steps in parallel, one  for each $f_i (i\in\{0,1\}^{k})$. 
\STATE Step 2:  Initialize {the} register  to the state $|\varphi_{{n-k}}\rangle=H^{\otimes n-k}|0^{n-k}\rangle$.
\STATE Step 3:  Execute the global Grover operator $G_2$ on  $|\varphi_{{n-k}}\rangle$ for $p_1$ {iterations}  to obtain $|\psi_{i_{n-k}}\rangle$. 
\STATE Step 4: Execute the local Grover operator $G_3$ on  $|\psi_{i_{n-k}}\rangle$ for $p_2$ {iterations} to obtain $|\psi_{i_{n-k}}^{\prime}\rangle$.
\STATE Step 5: Apply the generalized global Grover  operator $G_4$ on  $|\psi_{i_{n-k}}^{\prime}\rangle$ {once} to obtain $|\psi_{i_{n-k}}^{\prime\prime}\rangle$. 
\STATE Step 6: {Measuring} the first $p (p\le n-k, p\in N^{+})$ qubits yields the string $x_i\in\{0,1\}^{p}$, $ i\in\{0,1\}^{k}.$
\STATE Step 7: Based on the function $f_i$ and {the} bit string $x_i$,   construct a new function $f_{i,x_i}$   such that   $${f_{i,x_i}}(y)= {f_{i}}(x_iy), \forall y\in\{0,1\}^{n-p-k}, i\in\{0,1\}^{k}.$$ 
\STATE Step 8: {The}  $2^k$ machines will execute the following steps again in parallel,  one for each $f_{i,x_i}$.
\STATE Step 9:  Initialize {the} register  to the state $|\varphi_{{n-p-k}}\rangle=H^{\otimes n-p-k}|0^{n-p-k}\rangle$. 
 \STATE Step 10: Apply  {the} operator $L_i$ on  $|\varphi_{{n-p-k}}\rangle$ for $\lfloor(\pi/2-\theta')/(2\theta')\rfloor+1$ {iterations},  {and} then measure to produce $y_i\in\{0,1\}^{n-p-k}$.
\STATE Step 11: Verify whether $y_i$ is a solution of $f_{i,x_i}$ through a classical query of $y_i$. The target $ x_{ \widetilde{i}}y_t\widetilde{i}$ can be determined.  
\end{algorithmic}
\end{algorithm}

 Algorithm  \ref{algo21} is {referred to} as {the} IDGS algorithm. The IDGS algorithm consists of two stages: the first stage ({Steps} 1-6) addresses the first $p$ bits of the target using the exact quantum partial search algorithm, while the second stage ({Steps} 7-11) employs  the modified Grover’s algorithm  to find the remaining $n-p-k$ bits of the target.  The {circuits}  for the first and second stages of the IDGS algorithm are shown in Figs. \ref{ IDGS algorithm} and \ref{ IDGS algorithm2}. 

 \begin{figure}[htbp]
		\centering
		\begin{adjustbox}{width=1\textwidth}
\begin{quantikz}
\lstick{$\ket{0^{p}}$} & \gate[2]{H^{\otimes n-k}} \slice[shift=0, height=1.5]{$|\varphi_{n-k}\rangle$}\qw&  \gate[2]{G_{2}}\gategroup[wires=2,steps
=4,style={inner sep=-1.5pt},background,label style={label position=below,anchor=
north,yshift=-0.2cm}]{repeat $p_{1}$ times} & \qw  & \hspace{0.25cm} \ldots \hspace{0.5cm}  & \gate[2]{G_{2}}\slice{$|\psi_{i_{n-k}}\rangle$}   & \gate[2]{G_{3}}\gategroup[wires=2,steps
=4,style={inner sep=-1.5pt},background,label style={label position=below,anchor=
north,yshift=-0.2cm}]{repeat $p_{2}$ times}& \qw& \hspace{0.25cm} \ldots \hspace{0.5cm}& \gate[2]{G_{3}}\slice{$|\psi_{i_{n-k}}^{\prime}\rangle$} & \gate[2]{G_{4}} \slice{$|\psi_{i_{n-k}}^{\prime\prime}\rangle$}& \meter
{} & \rstick{$x_{i}$}\qw  \\
\lstick{$\ket{0^{n-p-k}}$}& \qw  & \qw& \qw& \hspace{0.25cm} \ldots \hspace{0.5cm}& \qw & \qw& \qw& \hspace{0.25cm} \ldots \hspace{0.5cm} & \qw& \qw& \qw&\qw&\qw
\end{quantikz}
	\end{adjustbox}
\caption{The circuit for the first stage of  {the} IDGS algorithm.}
		\label{ IDGS algorithm}
\end{figure}
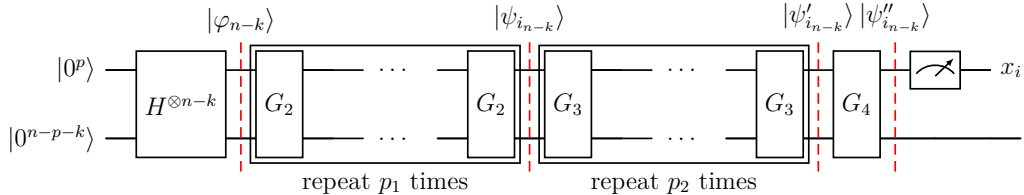
\begin{figure}[htbp]
		\centering
		\begin{adjustbox}{width=0.8\textwidth}
\begin{quantikz}
\lstick{$\ket{0^{n-p-k}}$} & \gate{H^{\otimes n-p-k}} \slice{}&  \gate{L_{i}}\gategroup[wires=1,steps
=3,style={inner sep=-1.5pt},background,label style={label position=below,anchor=
north,yshift=-0.2cm}]{repeat\small  $\lfloor(\pi/2-\theta')/(2\theta')\rfloor+1$ times}  & \hspace{0.25cm} \ldots \hspace{0.5cm}  & \gate{L_{i}} \slice{} & \meter
{} & \rstick{$y_{i}$} 
\end{quantikz}
	\end{adjustbox}
\caption{The circuit for the second stage of  {the} IDGS algorithm.}
		\label{ IDGS algorithm2}
\end{figure}
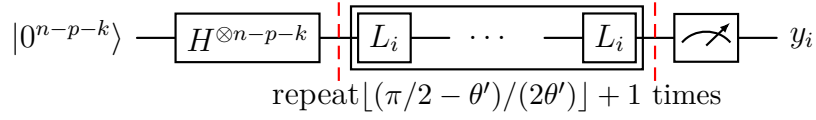
	
\begin {remark} 
When $p=n-k$, {the} IDGS algorithm is similar to the algorithm for single-objective cases in Ref.\cite{Qiu2022}, but our algorithm shows improvements in accuracy. In addition, if {there are} $a$ marked elements{, where $a>1$, }  and they are uniformly distributed across $T$ subfunctions, then the IDGS algorithm can be used to find the marked elements, resulting in more than {one} marked element being identified.
\end {remark}
\begin {remark} 
If $n-p-k$ is not very large (for example{,} $n-p-k\le4$),  when searching for the final $n-p-k$ bits at the end, we can choose to use a brute-force search instead of the modified Grover’s algorithm by Long. Additionally, Algorithm \ref{algo21} can also be executed {sequentially} by a single computer with $n-k$ qubits, which is equivalent to sacrificing time for space.
\end {remark}
\subsection{ Correctness analysis}
 In this section, we analyze the correctness of  {the} IDGS algorithm. First,  {we} analyze the target subfunction $f_{\widetilde{i}}$.  The block that includes the target is called the target block, while the one that does not include the target is called the non-target block.

In order to facilitate the subsequent  lemmas, we first introduce some notation. 
We denote   
  \begin{equation}\label{can}
a_{t} :=\sin \left(2 p_1\theta_1+\theta_1 \right)\cos \left(2 p_2\theta'\right)+\frac{\sqrt{2^{n-p-k}-1}\cos \left(2 p_1\theta_1+\theta_1\right) \sin \left(2 p_2\theta'\right)}{\sqrt{2^{n-k}-1}},\end{equation}
  \begin{equation}\label{can1}
a_{n_t} :=-\sin \left(2 p_1\theta_1+\theta_1\right)\sin \left(2 p_2\theta'\right)+\frac{\sqrt{2^{n-p-k}-1}\cos \left(2 p_1\theta_1+\theta_1\right) \cos \left(2 p_2\theta'\right)}{\sqrt{2^{n-k}-1}},\end{equation}
 where $\theta_1=\arcsin\frac{1}{\sqrt{2^{n-k}}}$, $\theta'= \arcsin\frac{1}{\sqrt {2^{n-p-k}}}$.

\begin{lemma}\label{l2}
After {Step} {\rm3} of Algorithm \ref{algo21}, the quantum state {$|\psi_{\widetilde{i}_{n-k}}\rangle$} can be expressed as {follows}:  \begin{align}\label{G2}
|\psi_{\widetilde{i}_{n-k}}\rangle=&|x_{ \widetilde{i}}\rangle\otimes\left(\sin \left(\left(2 p_1+1\right) \theta_1\right)|y_t\rangle+\sum_{x^{\prime}\neq y_t} \frac{\cos \left(\left(2 p_1+1\right) \theta_1\right)}{\sqrt{2^{n-k}-1}}|x^{\prime}\rangle\right)\\ \nonumber
&+\sum_{ \psi_{{p}}\neq x_{ \widetilde{i}}}|\psi_{{p}}\rangle\otimes \frac{\sqrt{2^{n-p-k}}\cos \left(\left(2 p_1+1\right) \theta_1\right)}{\sqrt{2^{n-k}-1}}|\varphi_{{n-p-k}}\rangle.
\end{align}
\end{lemma}

The proof of Lemma \ref{l2} is presented in  \ref{x5}. Next, we will demonstrate how to obtain $|\psi_{{\widetilde{i}}_{n-k}}^{\prime}\rangle$ through a lemma.
\begin {lemma}\label{l8}
After {Step} {\rm4} of  Algorithm \ref{algo21}, the  quantum state is 
\begin{align*}
|\psi_{{\widetilde{i}}_{n-k}}^{\prime}\rangle=&|x_{ \widetilde{i}}\rangle\otimes\left(a_{t}|y_t\rangle+a_{n_t}\sum_{x^{\prime}\neq y_t} \frac{|x^{\prime}\rangle}{\sqrt{2^{n-p-k}-1}}\right)\\
&+\sum_{ \psi_{{p}}\neq x_{ \widetilde{i}}}|\psi_{{p}}\rangle\otimes \frac{\sqrt{2^{n-p-k}}\cos \left(\left(2 p_1+1\right) \theta_1\right)}{\sqrt{2^{n-k}-1}}|\varphi_{{n-p-k}}\rangle.
\end{align*}
\end{lemma}
The proof of Lemma \ref{l8} is presented in  \ref{x6}.

We denote 
$$E :=\sqrt{2^{n-p-k}-1}\cdot {a_{n_{t}}}+\frac{\left(2^{n-k}-2^{n-p-k}\right)\cos \left(\left(2 p_1+1\right) \theta_1\right)}{\sqrt{2^{n-k}-1}},$$
$$F :=\frac{\cos \left(\left(2 p_1+1\right) \theta_1\right)}{\sqrt{2^{n-k}-1}}.$$
With  {this notation}, we  give the following lemma.
\begin {lemma}\label{720}
If $E$,  $F$ and ${a_{t}}$ satisfy the  inequality 
\begin{equation}\label{zong1}
\left(E-{2^{n-k-1}F}\right)^2\le{a_{t}}^2,
\end{equation}
let {$\phi$ and $\theta$ satisfy}\begin{equation}\label{xzong14} \left\{
\begin{aligned}
&\cos\theta=\frac{2^{2n-2k-1}F^2-2^{n-k}EF+E^2-a_{t}^2}{E^2-2^{n-k}EF-a_{t}^2}, \\
&\cos\phi=\frac{2^{n-k-1}F-E}{a_{t}},\\
&a_{t}\sin\phi\sin\theta>0.
\end{aligned}
\right.
\end{equation}
  {Then,}  $\phi$ and $\theta$   necessarily satisfy {the} {equations} \begin{equation}\label{xgzong14} \left\{
\begin{aligned}
\cos \left(\phi+\frac{\theta }{2}\right) & =  \frac{-E\cos\frac{\theta }{2}} {a_{t}}, \\
\sin\left(\phi+\frac{\theta}{2}\right)& =\frac{\left(\cos\theta-1 \right)E+2^{n-k}F}{2a_{t}\sin\frac{\theta}{2}}.
\end{aligned}
\right.
\end{equation}
\end{lemma}

The proof of the lemma is analogous to that of Lemma \ref{l3}, and will not be provided here.  Below, a lemma is presented to illustrate how the first  $p$ bits of the target can be precisely obtained. 

\begin{lemma}\label{t1}
 If  Eq. (\ref{xgzong14})  holds, 
 then we can obtain the first $p$ bits $x_{0}\cdots x_{p-1}$that satisfy {the} Boolean function $$f_{\widetilde{i}}(x_{0}\cdots x_{p-1}y_t)=1$$ with probability 1.
\end{lemma}
\begin{proof}
 {After Step} 5 of the algorithm, we have
\begin{equation}\label{zong12}
 \begin{aligned}
 G_4|\psi_{{\widetilde{i}}_{n-k}}^{\prime}\rangle=&\frac{{a_{t}}(1-e^{\imath\theta })e^{\imath\phi }}{\sqrt{2^{n-k}}}|\varphi_{{n-k}}\rangle-e^{\imath\phi }{a_{t}}|x_{ \widetilde{i}}y_t\rangle\\
&+\frac{{a_{n_{t}}}(2^{n-p-k}-1)(1-e^{\imath\theta })}{\sqrt{2^{n-p-k}-1}\sqrt{2^{n-k}}}|\varphi_{{n-k}}\rangle-|x_{ \widetilde{i}}\rangle\otimes\sum\limits_{x^{\prime}\neq y_t} \frac{a_{n_{t}}|x^{\prime}\rangle}{\sqrt{2^{n-p-k}-1}}\\
&+\cos \left(\left(2 p_1+1\right) \theta_1\right)\left[\frac{ (2^{n-k}-2^{n-p-k})(1-e^{\imath\theta })|\varphi_{{n-k}}\rangle}{\sqrt{2^{n-k}(2^{n-k}-1)}}\right.\\
&\left.-\sum_{ \psi_{{p}}\neq x_{ \widetilde{i}}}|\psi_{{p}}\rangle\otimes \frac{\sqrt{2^{n-p-k}}}{\sqrt{2^{n-k}-1}}|\varphi_{{n-p-k}}\rangle\right].
\end{aligned} 
\end{equation}
Since $$|\varphi_{{n-k}}\rangle=\frac{ |x_{ \widetilde{i}}y_t\rangle+|x_{ \widetilde{i}}\rangle\otimes\sum\limits_{x^{\prime}\neq y_t} |x^{\prime}\rangle+\sum_{ \psi_{{p}}\neq x_{ \widetilde{i}}}|\psi_{{p}}\rangle\otimes \sqrt{2^{n-p-k}}|\varphi_{{n-p-k}}\rangle}{\sqrt{2^{n-k}}},$$after the action of $G_4$, the amplitude of  {each element}  in each non-target block is 
\begin{align}\label{zong131}&\frac{{a_{t}}(1-e^{\imath\theta })e^{\imath\phi }}{2^{n-k}}+\frac{{a_{n_{t}}}(2^{n-p-k}-1)(1-e^{\imath\theta })}{2^{n-k}\sqrt{2^{n-p-k}-1}}\\ \nonumber
+&\frac{\cos \left(\left(2 p_1+1\right) \theta_1\right)(2^{n-k}-2^{n-p-k})(1-e^{\imath\theta })}{2^{n-k}\sqrt{2^{n-k}-1}}-\frac{\cos \left(\left(2 p_1+1\right) \theta_1\right)}{\sqrt{2^{n-k}-1}}.\end{align}
Combining $E$ and $F$ in Eq. (\ref{zong131}) yields the following equation{:}
\begin{align}\label{zong151}
&\frac{(1-e^{\imath\theta })({a_{t}}e^{\imath\phi }+E)}{2^{n-k}}-F \nonumber \\
=&\frac{(1-\cos \theta-\imath\sin\theta  )[{a_{t}}(\cos \phi+\imath\sin\phi )+E]}{2^{n-k}}-F \nonumber\\
=&\frac{(1-\cos \theta)({a_{t}}\cos \phi+E)+{a_{t}}\sin\theta\sin\phi }{2^{n-k}}-F\nonumber\\
&+\imath\frac{[{a_{t}}(1-\cos \theta)\sin\phi -{a_{t}}\sin\theta \cos \phi-E\sin\theta ]}{2^{n-k}}.
\end{align}
Finally, substituting Eq. (\ref{xgzong14}) into Eq. (\ref{zong151}) yields  $$\frac{(1-e^{\imath\theta })({a_{t}}e^{\imath\phi }+E)}{2^{n-k}}-F=0,$$i.e.,  the non-target block amplitude is 0. Hence, the target substring is finally obtained by measuring the first $p$ qubits. 
\end{proof}

Throughout the process, the number of executions of   Grover's  operator for non-target subfunctions is the same as that for the target subfunction, and all other parameter values are also the same. For the subfunction that does not contain the target, we have\begin{equation*}
 G_{2}^{p_1}|\varphi_{{n-k}}\rangle=|\varphi_{{n-k}}\rangle,
\end{equation*} 
\begin{equation*}
 G_{3}^{p_2}|\varphi_{{n-k}}\rangle=|\varphi_{{n-k}}\rangle,
\end{equation*} 
\begin{equation*}
 G_4|\varphi_{{n-k}}\rangle=-e^{\imath\theta }|\varphi_{{n-k}}\rangle.
\end{equation*} 
After {Step} 6, we {can also obtain} the $p$ bits  with  probability  $\frac{1}{2^{p}}$. 

In {Step} 9, the circuit is initialized to $ |\varphi_{{n-p-k}}\rangle$. Then, after applying the modified Grover’s algorithm by Long  in {Step} 10, {we} can always obtain an $n-p-k${-}bit string. 
 In the non-target subfunctions, applying the $L_i$ operator to $ |\varphi_{{n-p-k}}\rangle$ results in
$
 L_i |\varphi_{{n-p-k}}\rangle=-e^{\imath\omega
}|\varphi_{{n-p-k}}\rangle.
$ Therefore, by measuring, we randomly obtain the $n-p-k$ bits with  probability  $\frac{1}{2^{n-p-k}}$.

\begin{theorem}
  If  the first $p$ bits $x_{ \widetilde{i}}(x_{0}\cdots x_{p-1})$ are obtained with probability 1, then the target within an unsorted database can be precisely obtained through Algorithm \ref{algo21}.
\end{theorem}
\begin{proof}
Since one of the $f_i (i\in\{0,1\}^{k})$ is necessarily the target subfunction $f_{\widetilde{i}}$, {it follows from this condition that}  the target $x_{ \widetilde{i}}$ is obtained with probability 1. The non-target subfunctions yield $x_i(i\neq \widetilde{i}, i\in\{0,1\}^{k})$ with a probability of $\frac{1}{2^{p}}$. Thus, within  $2^k$ instances of $x_i$,  the first $p$ {bits} of  the target are assuredly contained. One of the ${f_{i,x_i}}(i\in\{0,1\}^{k})$ constructed from $x_i$ will also be the target subfunction. Utilizing the modified Grover’s algorithm by Long, we search all {blocks} in parallel, allowing us to obtain the remaining $n-p-k$  {bits} of the target with probability 1 and the $n-p-k$ {bits} of  {a} non-target with  probability $\frac{1}{2^{n-p-k}}$. 

Upon reaching {Step} 11, by substituting $y_i$ into ${f_{i,x_i}}$, one can ascertain that $y_t$ satisfies ${f_{i,x_i}}(y_t)=1$.
  Furthermore, since ${f_{\widetilde{i},x_i}}(y_t)={f_{\widetilde{i}}}(x_{ \widetilde{i}}y_t)$ and $f_{\widetilde{i}}(x)=f\left(x_{ \widetilde{i}}y_t\widetilde{i}\right)$,  the entire target is precisely acquired. 
\end{proof}

\subsection{Advantages of   {the} algorithm}
\begin{figure}[h]
		\centering
		\begin{minipage}[t]{0.48\linewidth}
		\begin{adjustbox}{width=0.8\textwidth}
			\begin{quantikz}[row sep={0.7cm,between origins},column sep=0.5cm]
					&  \gate{H} 	&  \gate{X} 	&  \ctrl{1}  &  \gate{X} 	&  \gate{H} 	&\qw \\
					&  \gate{H} 	&  \gate{X} &  \ctrl{2}  	 &  \gate{X} 	&  \gate{H}  &\qw \\[2em]
                    &  \gate{H} 	&  \gate{X} &  \ctrl{1}  	 &  \gate{X} 	&  \gate{H}  &\qw \\
					&  \gate{H} 	&  \gate{X} &  \gate{U(\theta)}      &  \gate{X} 	&  \gate{H}  &\qw 		
			\end{quantikz}
		
		\hspace{-5.45cm}
				\begin{tikzpicture}[>=stealth,baseline=0cm]
					\node at (0,0.1) {$\vdots$};
                    \node at (3.7,0.1) {$\vdots$};
				\end{tikzpicture}
		\end{adjustbox}
		\caption{ The circuit {for} $U_{n,\theta}${.}}
\label{fig1}
 \end{minipage}
		\begin{minipage}[t]{0.47\linewidth}
		\begin{adjustbox}{width=0.8\textwidth}
			\begin{quantikz}[row sep={0.7cm,between origins},column sep=0.5cm]
					  	&  \gate[][1.6cm]{X^{1-t_1}} 	&  \ctrl{1}  & \gate[][1.6cm]{X^{1-t_1}} 	 	&\qw \\
					  	&  \gate[][1.6cm]{X^{1-t_2}}  &  \ctrl{2}  	 &  \gate[][1.6cm]{X^{1-t_2}} 	  &\qw \\[2em]
                   	&  \gate[][1.6cm]{X^{1-t_{n-1}}}  &  \ctrl{1}  	 &  \gate[][1.6cm]{X^{1-t_{n-1}}}  	 &\qw \\
						&  \gate[][1.6cm]{X^{1-t_n}}  &  \gate{U(\phi)}      &  \gate[][1.6cm]{X^{1-t_n}} 	  &\qw 		
			\end{quantikz}
		
		\hspace{-5.38cm}
				\begin{tikzpicture}[>=stealth,baseline=0cm]
					\node at (0,0.1) {$\vdots$};
                    \node at (3.7,0.1) {$\vdots$};
				\end{tikzpicture}
		\end{adjustbox}
\caption{The circuit {for} $U_{f}^{\phi}${.}}
 \label{fig2}
 \end{minipage}
	\end{figure}
It is not difficult to see that our algorithm requires at most \( n-k \) qubits. According to Ref.\cite{Figgatt2017},  the circuits  {for}  $U_{n, \theta}=(1-e^{\imath\theta })|\varphi_{n}\rangle\langle\varphi_{n}|-I_{n}$ and $U_{f}^{\phi}$ are shown in \textbf{Figs. \ref{fig1}} and \textbf{ \ref{fig2}}, where $\theta, \phi\in \mathbb{R}$. The role of $U_{f}^{\phi}$ is to apply the phase $e^{\imath\phi}$ to the target  $|t\rangle=|t_1t_2\cdots t_{n}\rangle$, leaving non-target states unchanged.  There are different ways of constructing $U_{f}^{\phi}$ for different Boolean {functions} $f$. We take a simple approach {in Fig. 4}. In {Fig. 3 and Fig. 4}, $U(\theta)$ and $U(\phi)$ are defined as follows{:}\begin{equation*}U(\theta)=
\begin{pmatrix} 1 & 0 \\ 0 & e^{\imath\theta} \end{pmatrix},U(\phi)=
\begin{pmatrix} 1 & 0 \\ 0 & e^{\imath\phi} \end{pmatrix}.\end{equation*}

 The $U(\theta)$ with $n-1$ control qubits is denoted {by} $\wedge_{n}(U(\theta))$.   { $\wedge_{n}(U(\theta))$ is a {multi-controlled} {gate}. Consider its decomposition {into}  Clifford+T gates. According to Theorem 3 in Ref. \cite{Barenco1995}, the following lemma is obtained.} 
 
\begin {lemma} \label{fen1}
  $\wedge_{n}(U(\theta))(n\ge7)$ can be implemented using one ancilla qubit in {the} state \( |0\rangle \) and  a constant number of gates from \(\{R_{\hat{x}}, R_{\hat{z}}, H\}\),  {with the following costs and depths:}
\begin{enumerate}
    \item CNOT cost \( 12n - 36 \) and depth \( 8n-8 \),
    \item \( T \) cost \( 16n - 64 \) and depth \( 8n - 14 \) \rm{(}\( 8n - 11 \) for odd \( n-1 \)\rm{)},
    \item \( H \) cost \( 8n - 40 \) and depth \( 4n - 15 \),
\end{enumerate}
where \(\{R_{\hat{x}}, R_{\hat{z}}\}\) denote the single-qubit rotation gates about the \(x\)-axis and {the} \(z\)-axis.
\end{lemma}
{We define the depth of $\wedge_{n}(U(\theta))(n\ge7)$ as the maximum depth among all gates in the Clifford+T {gate}  set, i.e., $d(\wedge_{n}(U(\theta)))=\max(d(T), d(H), d(CNOT))$. Therefore, $d(\wedge_{n}(U(\theta)))=8n-8$.} Next, we will use a theorem to demonstrate the circuit depth of Algorithm \ref{algo21}.

\begin {theorem}  \label{fen2}
{ If $p\ge4$ {and} $n-p-k\ge7$},  then the circuit depth of  Algorithm \ref{algo21} is gradually decreasing. 
\end {theorem} 
\begin{proof}
We first calculate the circuit depths for executing \( G_{2} \), \( G_{3} \), and \( G_{4} \) once.

When $\theta=\phi=\pi$, we have $U(\theta)=U(\phi)=Z$, and $U_{n-k, \theta}$ represents the circuit {for} $U_{{n-k, \pi}}$, while $U_{f}^{\phi}$ represents the circuit {for} $U_{f}^\pi$.
According to Lemma \ref{fen1}, we derive that  the circuit depth of $U_{f_{i}}^\pi$ is $$d(U_{f_{i}}^\pi)=2+d(\wedge_{n-k}(U(\pi)))=8(n-k)-6.$$ The circuit depth {of} $U_{{n-k, \pi}}$ is  $$d(U_{{n-k, \pi}})=4+d(\wedge_{n-k}(U(\pi)))=8(n-k)-4.$$ 
It can be further concluded that \begin{equation}\label{h1}d(G_{2}) =d(U_{{n-k, \pi}})+d(U_{f_{i}}^\pi)=16{(n-k)}-10;\end{equation}
\begin{equation}\label{h2}d(G_{3}) =d(U_{{n-p-k, \pi}})+d(U_{f_{i}}^\pi)=16(n-k)-8p-10.\end{equation}
We have $$d(G_{2})-d(G_{3})=8p>0.$$
  The circuit depth of $L_i$ is \begin{equation}\label{h3}d(L_i) =d(U_{{n-p-k, \omega}})+d(U_{f_{i,x_i}}^{\omega})=16{(n-k-p)}-10.\end{equation}
 Furthermore, {since} $$d(G_{3})-d(L_i)=8p>0,$$
 we have $d(G_{2})=d(G_{4})>d(G_{3})>d(L_i).$
{Therefore, the circuit depth is given by}
\begin{align}\label{zong1516}
&p_{1}d(G_{2})+p_{2}d(G_{3})+d(G_{4})-\left(\lfloor(\pi/2-\theta')/(2\theta')\rfloor+1\right)d(L_i) \nonumber \\
&>\left(p_{1}+p_{2}-\lfloor(\pi/2-\theta')/(2\theta')\rfloor\right)d(L_i)\nonumber \\
&>\left(p_{1}+p_{2}-\frac{\pi}{4\theta'}\right)d(L_i) \nonumber \\
&>\left[\left(\frac \pi 4+\frac{\eta_p- \gamma_p}{\sqrt {2^{p}}}-\frac{\pi}{4\sqrt {2^{p}}}\right)\sqrt {2^{n-k}}-1\right]d(L_i){.}
\end{align}
 The last inequality is obtained using $x-\frac{1}{2} \le \lfloor x \rceil$ and $\theta'=\arcsin{\frac {1} {\sqrt{2^{n-p-k}}}}\approx\frac {1} {\sqrt{2^{n-p-k}}}$. 
 
 Next, we will focus on analyzing $\frac \pi 4+\frac{\eta_p- \gamma_p}{\sqrt {2^{p}}}$. Let 
\begin{equation}\label{zong1512}
f(x)=\frac \pi 4+\frac{\arcsin\left(\frac{\sqrt {x}}{\sqrt{4(x-1)}}\right)}{\sqrt {x}}-\frac{1}{2}\arctan\left(\frac{\sqrt{3 x-4}}{x-2}\right), x\ge16.
\end{equation}
The derivative of the function $f(x)$ is given by
 \begin{equation}\label{zong1512}
f^{\prime}(x)=\frac{\sqrt{3 x-4}}{4x(x-1)}-\frac{\arcsin\left(\frac{1}{2\sqrt{1-\frac{1}{x}}}\right)}{2x^{\frac{3}{2}}}.
\end{equation}
Simplifying Eq. (\ref{zong1512}), we obtain 
\begin{equation}\label{zong1513}
f^{\prime}(x)=\frac{x^{\frac{1}{2}}\sqrt{3x-4}-2(x-1)\arcsin\left(\frac{1}{2\sqrt{1-\frac{1}{x}}}\right)}{4x^{\frac{3}{2}}(x-1)}.
\end{equation}
Since $x>2$, it follows that $\frac{1}{2}<\frac{1}{2\sqrt{1-\frac{1}{x}}}<\frac{1}{\sqrt{2}}$. Therefore, $\frac{\pi}{6}<\arcsin\left(\frac{1}{2\sqrt{1-\frac{1}{x}}}\right) < \frac{\pi}{4}$.
Thus, we have $$x^{\frac{1}{2}}\sqrt{3x-4}-2(x-1)\arcsin\left(\frac{1}{2\sqrt{1-\frac{1}{x}}}\right)>x^{\frac{1}{2}}\sqrt{3x-4}-\frac \pi 2 (x-1)>0.$$
From this, we know that $f^{\prime}(x)> 0$, which means $f(x)$ is an increasing function. As a result, the minimum value of  $\frac \pi 4+\frac{\eta_p- \gamma_p}{\sqrt {2^{p}}}$ occurs at {$p=4$}. For {$p\ge4$, $\frac \pi 4+\frac{\eta_p- \gamma_p}{\sqrt {2^{p}}}\ge 0.699$. Moreover, since $\frac \pi 4+\frac{\eta_p- \gamma_p}{\sqrt {2^{p}}}-\frac{\pi}{4\sqrt {2^{p}}}\ge 0.699-\frac {\pi} {16}>0.5$,} it follows that $$\left(\frac \pi 4+\frac{\eta_p- \gamma_p}{\sqrt {2^{p}}}-\frac{\pi}{4\sqrt {2^{p}}}\right)\sqrt {2^{n-k}}-1>0.$$
Ultimately, we can obtain \begin{equation}\label{zong16}
p_{1}d(G_{2})+p_{2}d(G_{3})+d(G_{4})>\left(\lfloor(\pi/2-\theta')/(2\theta')\rfloor+1\right)d(L_i).\end{equation}
In summary, it can be concluded that the circuit depth is gradually decreasing.
\end{proof}

\begin {remark} 
 Algorithm \ref{algo21} can be further extended to multiple stages. Specifically, by setting \( n-k = \sum_{i'=1}^{m'} p_{i' } \), we can iteratively apply {Steps}  1 to 6 to compute ${p_{i'}}(i'\in\{ 1, \cdots, m'-1\})$ {bits} at each stage, which significantly reduces the circuit depth.
\end {remark}
According to Theorem \ref{fen2}, our maximum circuit depth primarily depends on the first stage of {the} IDGS algorithm. Next, we consider the trade-off between the number of nodes and the total query complexity \cite{De2002}. The total query complexity is defined as the sum of the query complexities of individual nodes. Let {us} first look at the query complexity of a single node. The query complexity of {the} IDGS algorithm in the first phase is $$ p_{1}+p_{2}+1\approx\lfloor\frac \pi 4\sqrt{ 2^{n-k}}-0.34\sqrt{ 2^{n-p-k}}\rfloor+1.$$  The query complexity in the second phase is $$\left(\lfloor(\pi/2-\theta')/(2\theta')\rfloor+1\right)\approx\lfloor\frac \pi 4\sqrt{ 2^{n-p-k}}\rfloor +1.$$ Therefore, the query complexity of {the} IDGS algorithm  for a single node is \begin{align*}&\lfloor\frac \pi 4\sqrt{ 2^{n-k}}-0.34\sqrt{ 2^{n-p-k}}\rfloor+1+ \lfloor\frac \pi 4\sqrt{ 2^{n-p-k}}\rfloor +1\\
\approx&\frac \pi 4\sqrt{ 2^{n-k}}+0.45\sqrt{ 2^{n-p-k}}+2.\end{align*}
In calculating the query complexity of a single node, we ignore the error introduced by rounding. 
The total query complexity is $$ \frac \pi 4\sqrt{ 2^{n+k}}+0.45\sqrt{ 2^{n-p+k}}+2^{k+1}.$$

In the single-objective case, we use  \textbf{Table \ref{tabel1}} to present the comparison between our distributed algorithm and other algorithms in Refs.\cite{Grover1996, Long2001, Qiu2022,  Avron2021, LiH2024}. In Table 1, $d (U_{{n}}U_{f}^{\phi})$ represents the circuit depth of a single query by the algorithm, where the subscript $n$ indicates the number of qubits, and 
$\phi$ represents the angle of rotation towards the target.    {The total query complexity for each algorithm in Table 1 is obtained by calculating the number of calls to $U_{f}$.} The parameter $c$ in  Table 1 is  $c= \gamma_p-\eta_p>0$. Numerical experiments show that $c$ is approximately 0.34.
From Table 1, it is clear that our algorithm shows  advantages in circuit depth.

\begin{table}[h] 
	\centering
	\caption{{Comparison} of our {algorithm} with other  algorithms in \cite{Grover1996, Long2001, Qiu2022,  Avron2021, LiH2024}.}
	\resizebox{.9\columnwidth}{!}{
	\begin{tabular}{*{6}{c}}
		\toprule
		          Algorithms      &  {Number of qubits} & \makecell[c]{ Accurate}
		                 &  \makecell[c]{Quantum \\communication \\ complexity}&  {Circuit depth }&  {Total query complexity}\\
		\hline
		  \makecell[c]{ Grover's \\algorithm\cite{Grover1996} }      &       $n$   & \makecell[c]{inaccurate}  
		  & 0 &{$\lfloor\frac \pi 4\sqrt{ 2^{n}}\rfloor d(U_{{n}}U_{f}^{\pi})$}  & {$\lfloor\frac \pi 4\sqrt{ 2^{n}}\rfloor$}\\
		  
		  \makecell[c]{Modified \\Grover's algorithm \\ by Long\cite{Long2001}}         &      $n$ & accurate      &    0  &\makecell[c]{$\left(\lfloor\frac \pi 4\sqrt{ 2^{n}}\rfloor +1\right) d(U_{{n}}U_{f}^{\phi})$} & {$\lfloor\frac \pi 4\sqrt{ 2^{n}}\rfloor +1$}\\

 \makecell[c]{The algorithm \\  by Qiu et al\cite{Qiu2022} }      &        $n-k$ & \makecell[c]{inaccurate}  
		  & 0&{$\lfloor\frac \pi 4\sqrt{ 2^{n-k}}\rfloor d(U_{{n-k}}U_{f}^{\pi})$}  & {$2^{k}\lfloor\frac \pi 4\sqrt{ 2^{n-k}}\rfloor$}\\

  \makecell[c]{The algorithm\\ by Li et al\cite{LiH2024}}     &   \makecell[c]{$\max\left\{
\begin{array}{l}
n-k+1, \\
2^k + k + 1
\end{array}
\right\}$ } & \makecell[c]{ accurate}       
		  &    $O\left(n^2\sqrt{2^n}\right)$   & \makecell[c]{Greater than\\ $\left(\lfloor\frac \pi 4\sqrt{ 2^{n-k}}\rfloor +1\right) d(U_{{n-k+1}}U_{f}^{\phi})$}         & {$2^{k+1}\lfloor\frac \pi 4\sqrt{ 2^{n}}\rfloor$}\\[0.3cm]
		  
		  \makecell[c]{The algorithm \\  by Avron et al\cite{Avron2021} }      &        $n-1$ & \makecell[c]{inaccurate}   
		  & 0 &{$\lfloor\frac \pi 4\sqrt{ 2^{n-1}}\rfloor d(U_{{n-1}}U_{f}^{\pi})$}  & {2$\lfloor\frac \pi 4\sqrt{ 2^{n-1}}\rfloor$}\\
		  
		  \makecell[c]{ Our Algorithm \ref{algo21}\\   }      &        $n-k$ & \makecell[c]{ accurate}   
		  & 0 &\makecell[c]{$\approx\lfloor\frac \pi 4\sqrt{ 2^{n-k}}-c\sqrt{ 2^{n-p-k}}\rfloor d(U_{{n-k}}U_{f}^{\phi})$}  & {$ \frac \pi 4\sqrt{ 2^{n+k}}+0.45\sqrt{ 2^{n-p+k}}+2^{k+1}$}\\ [0.3cm]
		\toprule
	\end{tabular}\label{tabel1}
	}
\end{table}

\section{Numerical {experiments}}\label{s5}
In this section, we utilize MindQuantum to conduct   {simulation experiments}, thereby further illustrating the correctness of our algorithm. After introducing noise into the quantum system, we compare the IDGS algorithm with the modified Grover’s algorithm by Long\cite{Long2001} and find that the IDGS algorithm has noise resistance capability.

\subsection{5-qubit system without noise}
Given {a} Boolean function $f:\{0,1\}^5 \rightarrow \{0,1\}$ with $f(01100) = 1$, let \( p= 2 \) and \( k = 1 \). We find a two-bit substring each time. The function \( f \) is thereby divided into two subfunctions: \( f_0(x) = f(x0) \) and \( f_1(x) = f(x1) \), $x\in\{0,1\}^4$. Since the target $t=01100$, the target subfunction is $f_0$, and the non-target subfunction is $f_1$. We run Algorithm \ref{algo21} in each of these two subfunctions. Initially, we  find the first two  {bits} within \( f_i \)$(i=0,1)$.

Here, we need to clarify one point: the circuit in  Fig. \ref{ IDGS algorithm} measures the  {first} $p$ bits, but the implementation of the circuit is reversed.  This is because, in MindQuantum, the order of the bits in {the} register is usually opposite to the numbering of the qubits (for example, the measurement result "011" corresponds to $|q_2\rangle = 0, |q_1\rangle = 1, |q_0\rangle = 1)$.

According to Eq. (\ref{cqu1}), it follows that \( p_1 \approx \lfloor 1.2309 \rceil=1 \) and \( p_2 \approx \lfloor 1.2309 \rceil=1 \). The two rotational angles  are $\theta_1=\arcsin{\frac {1}{4}}$ and $\theta'=\arcsin{\frac {1}{2}}$. From Lemma \ref{l8}, we obtain the {amplitudes} as follows{:} \begin{equation*}\label{can1}
\begin{aligned}
&{a_{t}}=\sin \left(3\arcsin\frac{1}{4} \right)\cos \left(2 \arcsin\frac{1}{2}\right)+\frac{\cos \left(3\arcsin\frac{1}{4}\right) \sin \left(2 \arcsin\frac{1}{2}\right)}{\sqrt{5}}>0,\\
&{a_{n_t}}=-\sin \left(3\arcsin\frac{1}{4} \right)\sin \left(2 \arcsin\frac{1}{2}\right)+\frac{\cos \left(3\arcsin\frac{1}{4}\right) \cos \left(2 \arcsin\frac{1}{2}\right)}{\sqrt{5}}.\end{aligned} \end{equation*}
{The} parameter values for $\theta$ and $\phi$ are obtained from Lemma \ref{720}, i.e., 
\begin{align*}
&\theta=\pm\arccos\frac{2^{2n-2k-1}F^2-2^{n-k}EF+E^2-a_{t}^2}{E^2-2^{n-k}EF-a_{t}^2}\approx\pm2.3520, \\
&\phi=\pm\arccos\left(\frac{2^{n-k-1}F-E}{a_{t}}\right)\approx\pm1.5708,
\end{align*}
where $$E=\frac{3}{2}, F=\frac{3}{\sqrt{16}}.$$
Thus, the parameters for {Step} 5 are $\theta=2.3520$ and $\phi=1.5708$ or $\theta=-2.3520$ and $\phi=-1.5708$. With {the} parameters $\theta$ and $\phi$, we will next execute  Algorithm \ref{algo21} {in} MindQuantum, assuming that both parameters are positive numbers.

Consequently, the circuits of the search algorithm in the first phase are shown in \textbf{Figs. \ref{p}} and \textbf{ \ref{pa}}. It is worth noting that, since $f_1$ is a non-target subfunction, $U_ {f_1}$ acts as the identity operator. In the design of the circuit, we leverage this property to ensure optimal functionality.
\begin{figure*}[htbp]
		\centering
		\includegraphics[width=\linewidth]{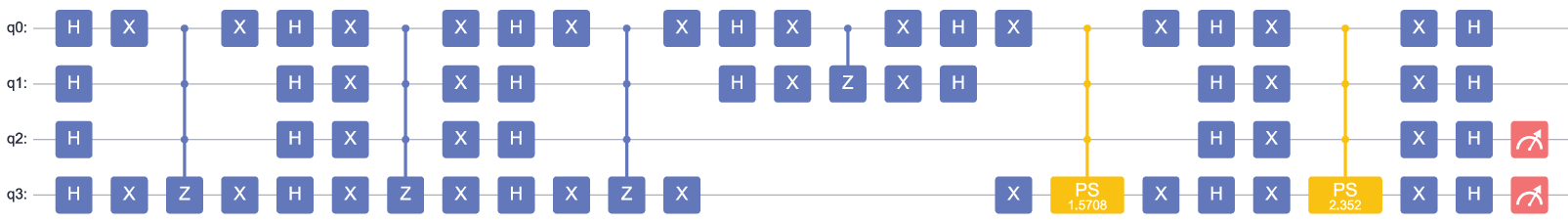}
		\setlength{\abovecaptionskip}{-0.01cm}
		\caption{The circuit for  $f_0$ in the first phase of {the}  IDGS algorithm.}
		\label{p}
	\end{figure*}
	\begin{figure*}[htbp]
		\centering
		\includegraphics[width=5.5in]{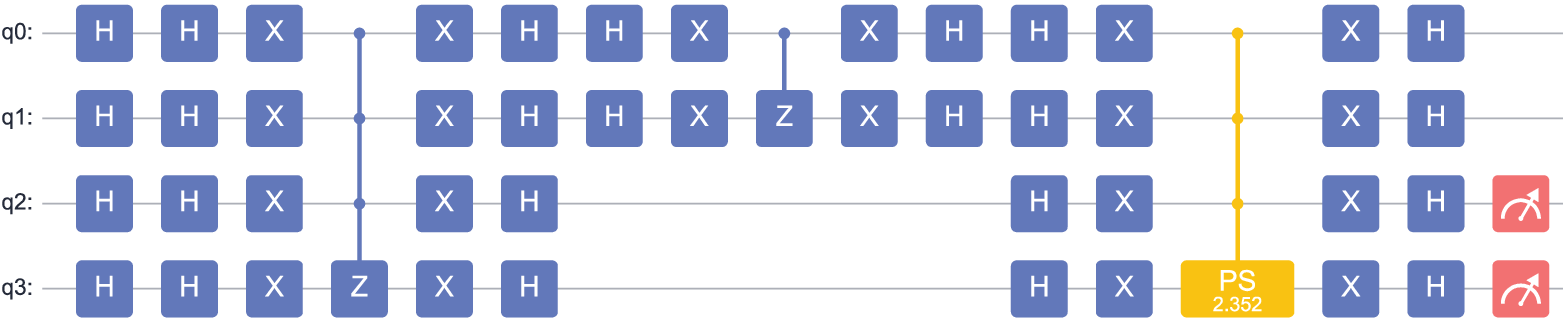}
		\setlength{\abovecaptionskip}{-0.01cm}
		\caption{The circuit for  $f_1$ in the first phase of {the}  IDGS algorithm.}
		\label{pa}
	\end{figure*}
By sampling the circuit  in  {Fig. 5}, we obtain the first two bits of the target, which are $01$, with the sampling results shown in \textbf{Fig. \ref{p2}}. Sampling the circuit for the non-target subfunction $f_1$ yields \textbf{Fig. \ref{p3}}. From this figure, it is evident that each string is equally probable. Hence, we can assume, for simplicity, that the outcome of the measurement is $10$. 
\begin{figure*}[htbp]
		\centering
		\includegraphics[width=4in]{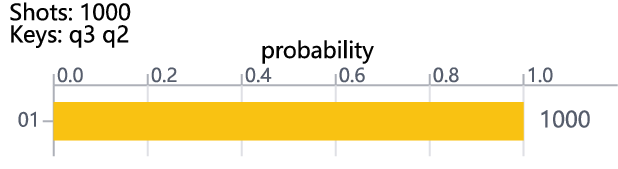}
		\setlength{\abovecaptionskip}{-0.01cm}
		\caption{Sampling results  of the circuit in {Fig. 5}.}
		\label{p2}
	\end{figure*}
\begin{figure*}[htbp]
		\centering
		\includegraphics[width=4in]{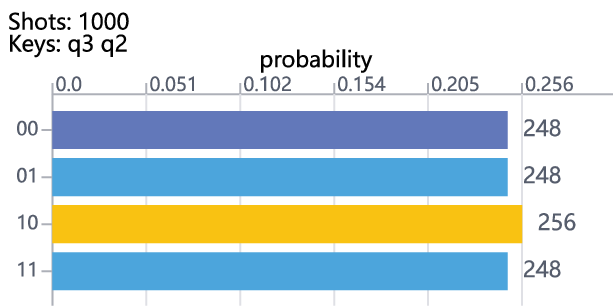}
		\setlength{\abovecaptionskip}{-0.01cm}
		\caption{Sampling results  of the circuit in {Fig. 6}.}
		\label{p3}
	\end{figure*}

Utilizing the information obtained from previous measurements, we further construct the Boolean functions $f_{0,01}:\{0,1\}^2 \rightarrow \{0,1\}$ and $f_{1,10}:\{0,1\}^2 \rightarrow \{0,1\}$, where  ${f_{0,01}}(y)= {f_{0}}(01y), {f_{1,10}}(y)= {f_{1}}(10y)$, and $y\in\{0,1\}^2.$
 {We then use} the modified Grover’s algorithm by Long to find the remaining target strings. {The corresponding} circuits are shown in \textbf{Figs. \ref{p9}} and \textbf{ \ref{p10}}. The values of the rotation angle parameters in {Fig. 9 and Fig. 10} are as follows:$$\omega=2\arcsin\left(\sin\left(\frac{\pi}{4\lfloor(\pi/2-\theta')/(2\theta')\rfloor+6}\right) / \sin  \theta'\right)\approx3.1416,$$
where $\theta'=\arcsin{\frac {1}{2}}$. Similarly, by sampling the circuit  {in}  {Fig. 9} , we {obtain} the last two bits of the target, namely 10, with the sampling results shown in \textbf{Fig. \ref{p11}}. By sampling  the circuit {for} the non-target subfunction {$f_{1,10}$}, we obtain \textbf{Fig. \ref{p12}}. {Fig. 12} indicates that we obtain each quantum state {with equal probability},  so it is reasonable to assume that the state obtained from the measurement is 10.
\begin{figure*}[htbp]
		\centering
		\includegraphics[width=\linewidth]{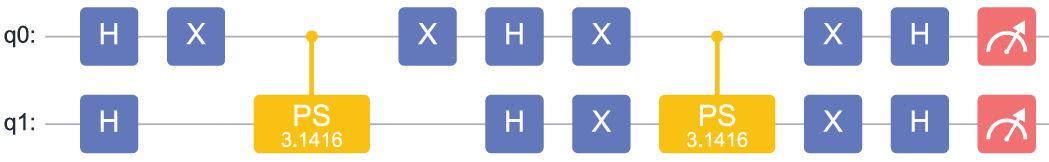}
		\setlength{\abovecaptionskip}{-0.01cm}
		\caption{The circuit for  $f_{0,01}$ in the final phase of  {the} IDGS algorithm.}
		\label{p9}
	\end{figure*}
	
	\begin{figure*}[htbp]
		\centering
		\includegraphics[width=6in]{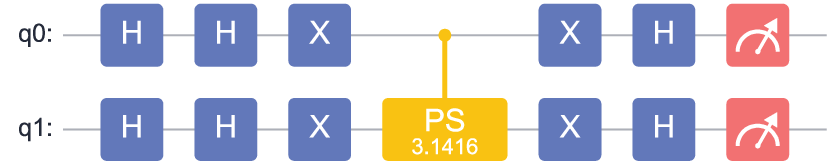}
		\setlength{\abovecaptionskip}{-0.01cm}
		\caption{The circuit for  $f_{1,10}$ in the final phase of  {the} IDGS algorithm.}
		\label{p10}
	\end{figure*}
	
	\begin{figure*}[htbp]
		\centering
		\includegraphics[width=4in]{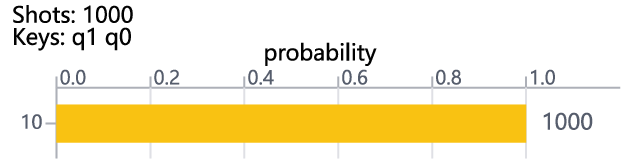}
		\setlength{\abovecaptionskip}{-0.01cm}
		\caption{Sampling results of the circuit in {Fig. 9}.}
		\label{p11}
	\end{figure*}
	
	\begin{figure*}[htbp]
		\centering
		\includegraphics[width=4in]{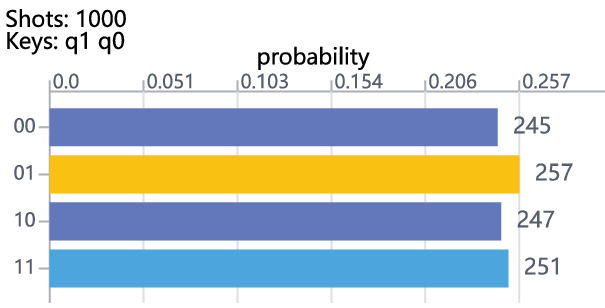}
		\setlength{\abovecaptionskip}{-0.01cm}
		\caption{Sampling results  of the circuit in {Fig. 10}.}
		\label{p12}
	\end{figure*}
Ultimately, by substituting the result $10$ into $f_{0,01}$ for verification, we find that $f_{0,01}(10)=1$ {and}  $f_{1,10}(10)=0$, thereby precisely obtaining the target $0110$ within the target subfunction $f_0$. Therefore, the target $01100$ in $f$ is successfully found.

\subsection{12-qubit system without noise}
Next, we consider larger-scale noise-free experiments. 
Given {a} Boolean function $f:\{0,1\}^{12} \rightarrow \{0,1\}$ with $f(111000001111) = 1$, let \( p= 3 \) and \( k = 1 \). The function \( f \) is  divided into two subfunctions: \( f_0(x) = f(x0) \) and \( f_1(x) = f(x1) \), $x\in\{0,1\}^{11}$. Since the target $t=111000001111$, the target subfunction is $f_1$, and the non-target subfunction is $f_0$. We run Algorithm \ref{algo21} in each of these two subfunctions. We first search the first three bits.

Based on the analysis of Algorithm \ref{algo21}, we first give the values of the parameters involved. According to Eq. (\ref{cqu1}), it follows that \( p_1 \approx \lfloor 21.0497 \rceil=21 \) and \( p_2 \approx \lfloor 9.0231 \rceil=9 \). The two rotational angles  are $\theta_1=\arcsin{\frac {1}{32\sqrt{2}}}$ and $\theta'=\arcsin{\frac {1}{16}}$. From Lemma \ref{l8}, we obtain the amplitude ${a_{t}}>0$. 
{The} parameter values for $\theta$ and $\phi$ are obtained from Lemma \ref{720}, i.e., 
\begin{equation*}
\theta\approx\pm3.0962, \quad
\phi\approx\pm0.5911.
\end{equation*}
Thus, the parameters for {Step} 5 are $\theta=3.0962$ and $\phi=0.5911$ or $\theta=-3.0962$ and $\phi=-0.5911$. After both parameters take positive values, we execute the algorithm on MindQuantum.

Since the detailed {circuits}  for the two subfunctions under Algorithm \ref{algo21} in the small-scale case have been provided in the previous section, for the larger-scale scenario, although the design method is exactly the same, the complete circuit is too extensive. Therefore, we directly present the sampling {results} for the target subfunction under Algorithm \ref{algo21} here. {These} results are shown in \textbf{Figs. \ref{p21} }and \textbf{\ref{p22}}, which once again verify the accuracy of our algorithm.
\begin{figure*}[htbp]
		\centering
		\includegraphics[width=4in]{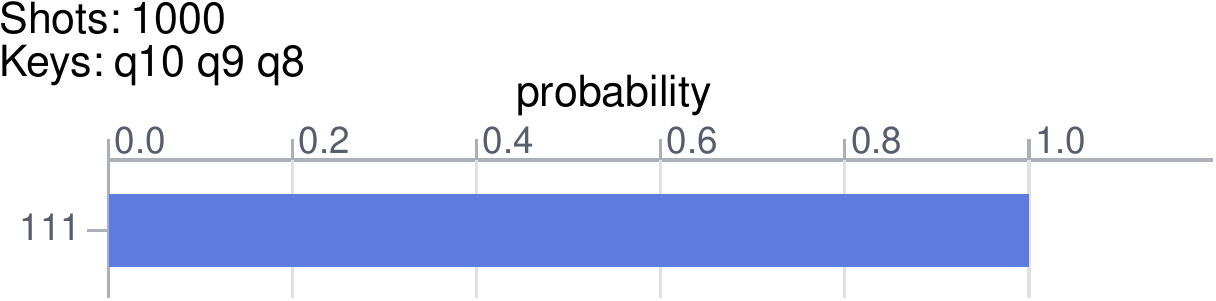}
		\setlength{\abovecaptionskip}{-0.01cm}
		\caption{Sampling results  of the circuit for $f_1$ in the first phase of  {the} IDGS algorithm.}
		\label{p21}
	\end{figure*}
	\begin{figure*}[htbp]
		\centering
		\includegraphics[width=4in]{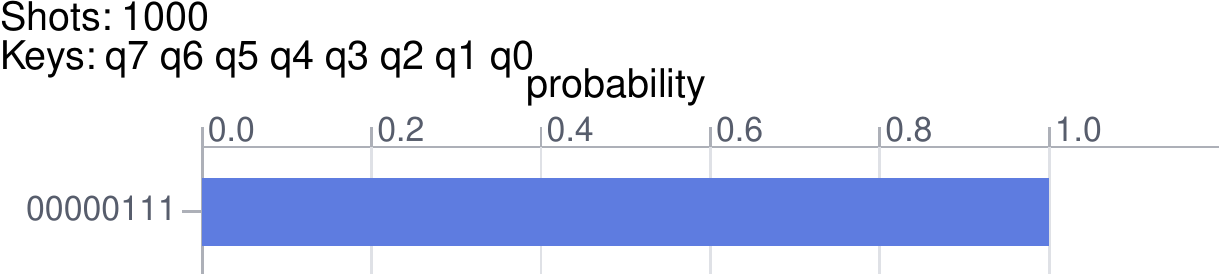}
		\setlength{\abovecaptionskip}{-0.01cm}
		\caption{Sampling results  of the circuit for $f_{1,111}$ in the final phase of  {the} IDGS algorithm. .}
		\label{p22}
	\end{figure*}
	
Ultimately, by substituting the result $00000111$ into $f_{1,111}$ for verification, we find that $f_{1,111}(00000111)=1$, thereby precisely obtaining the target $11100000111$ within the target subfunction $f_1$. Therefore, the target $111000001111$ in $f$ is successfully found.

Now consider the circuit depth  of our algorithm at this scale. Given that $n=12, k=1, p=3$, according to Eqs. (\ref{h1}), (\ref{h2}), and (\ref{h3}), we have $$d(G_{2})=166, d(G_{3})=142, d(L_i)=118.$$
Thus, the circuit depth of the first stage is $$166\times21+142\times9+166=4930.$$ Since the query complexity of the final stage is $$\lfloor\frac \pi 4\sqrt{ 2^{n-p-k}}\rfloor +1=13,$$ the circuit depth of the final stage is $118\times13=1534$. Therefore, the circuit depth of our algorithm is $4930.$
 Next, {we} analyze the circuit depth of Grover's algorithm. Since the query complexity  is $$\lfloor\frac \pi 4\sqrt{ 2^{n}}\rfloor =49,$$
 using Eq. (\ref{h1}), we have $182\times49=8918.$
Compared with the circuit depth of our algorithm, the circuit depth of Grover's algorithm increases by 3988.

\subsection{5-qubit system under the amplitude damping channel}
In practice, quantum systems are not closed, {and} they are subject to environmental interference. Therefore, when designing algorithms, it is usually necessary to consider their robustness. Consequently, based on the previous experiment with the 5-qubit system, we introduce amplitude damping noise into the experiment to test the robustness of {the}   algorithm.

The amplitude damping channel describes the dissipation of energy in a quantum system\cite{Nielsen2001}.  For any quantum state $\rho$,  if it undergoes the amplitude damping channel $\Phi$, then it {is transformed as follows:} 
$$
\Phi (\rho )=\sum_{i=0}^{1} E_i \rho E_i^{\dagger}, 
$$
where
\[
E_0=\begin{pmatrix}
  1 & 0  \\
  0 &  \sqrt{1-\gamma}
\end{pmatrix}, E_1=\begin{pmatrix}
  0 & \sqrt{\gamma}  \\
  0 & 0
\end{pmatrix},
\]
 $E_i$ are Kraus operators satisfying $\sum_{i=0}^{1}E_i^\dag
E_i=I$ with $I$ being the identity operator, and $\gamma$ is the dissipation coefficient with $0<\gamma<1$. After passing through the amplitude damping channel, the energy of the qubit dissipates, causing a reduction in the component of the $|1\rangle $ state in the superposition state. 

We introduce noise after each quantum operation. Here, we provide only the noisy circuit simulations conducted for the target subfunction $f_0$. The scenarios involving the addition of noise to the circuits corresponding to $f_0$ and $f_{0,01}$ in {Figs. 5 and 9} are illustrated in \textbf{Figs. \ref{p13}} and \textbf{\ref{p14}}.  The {results} of 1000 measurements on the circuits shown in {Figs. 15 and 16}, with the dissipation coefficient $\gamma=0.02$, are depicted in \textbf{Figs. \ref{p15}} and \textbf{\ref{p16}}. The algorithm's success rate is determined by the equation \begin{equation}\label{ss1}
p = \bar{p_1}\bar{p_2},\end{equation}where $\bar{p_1} $ represents the success rate of the first phase, and $\bar{p_2}$ denotes the success rate of the second phase. Therefore, based on the results shown in {Figs. 17 and 18}, we can calculate the success probability of the algorithm  as $p=0.653\times0.887=0.579211.$

\begin{figure*}[htbp]
		\centering
		\includegraphics[width=\linewidth]{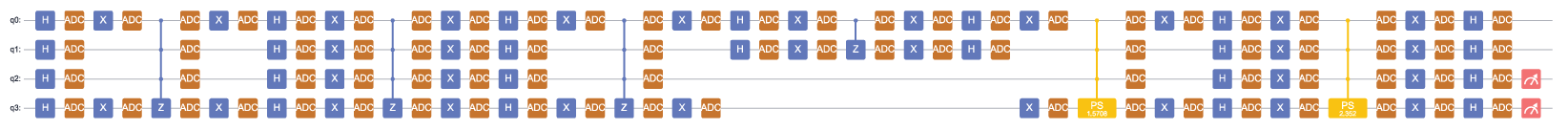}
		\setlength{\abovecaptionskip}{-0.01cm}
		\caption{The circuit for  $f_0$ with noise in the first phase of  {the} IDGS algorithm.}
		\label{p13}
	\end{figure*}
	\begin{figure*}[htbp]
		\centering
		\includegraphics[width=\linewidth]{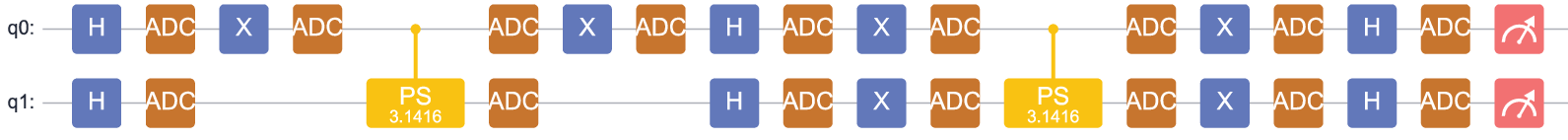}
		\setlength{\abovecaptionskip}{-0.01cm}
		\caption{The circuit for  $f_{0,01}$ with noise in the final phase of  {the} IDGS algorithm. }
		\label{p14}
	\end{figure*}

\begin{figure*}[htbp]
		\centering
		\includegraphics[width=4in]{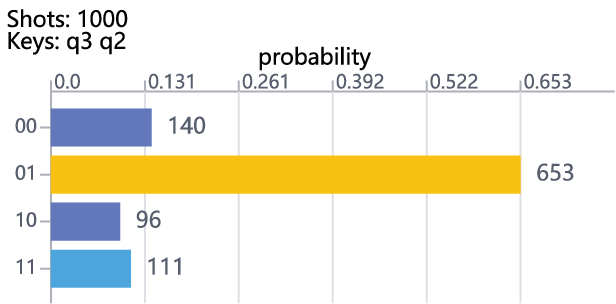}
		\setlength{\abovecaptionskip}{-0.01cm}
		\caption{Sampling results  of the circuit in {Fig. 15}.}
		\label{p15}
	\end{figure*}
\begin{figure*}[htbp]
		\centering
		\includegraphics[width=4in]{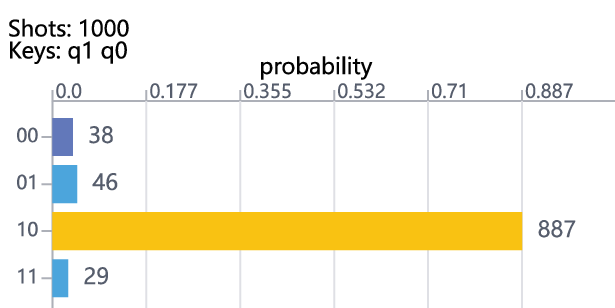}
		\setlength{\abovecaptionskip}{-0.01cm}
		\caption{Sampling results of the circuit in {Fig. 16}.}
		\label{p16}
	\end{figure*}
	
In the modified Grover’s algorithm by Long, amplitude damping noise is also incorporated, with its circuit  depicted in \textbf{Fig. \ref{p17}}. The value of the rotation angle parameter in {Fig. 19} is as follows:$$\omega=2\arcsin\left(\sin\left(\frac{\pi}{4\lfloor(\pi/2-\theta')/(2\theta')\rfloor+6}\right) / \sin  \theta'\right)\approx2.7648,$$
where $\theta'=\arcsin{\sqrt{\frac {1}{2^5}}}$.
Similarly, by setting the dissipation coefficient $\gamma=0.02$,  the outcomes of 1000 measurements {on} the circuit shown in {Fig. 19} are displayed in \textbf{Fig. \ref{p18}}. From the data presented in {Fig. 20}, it is observed that the success rate of the modified Grover’s algorithm is merely 29.4$\%$. To more effectively compare the noise resistance capabilities of the two algorithms, experiments are performed using a range of dissipation coefficients $\gamma\in\{0.001,0.005,0.007,0.010,0.020,0.030,0.040,0.050\}$. The outcomes of these experiments are depicted in \textbf{Fig. \ref{p19}}.  {Fig. 21} indicates that our algorithm maintains a relatively high success probability even under conditions of substantial noise, thereby demonstrating a certain degree of noise resilience inherent to the IDGS algorithm.
\begin{figure*}[htbp]
		\centering
		\includegraphics[width=\linewidth]{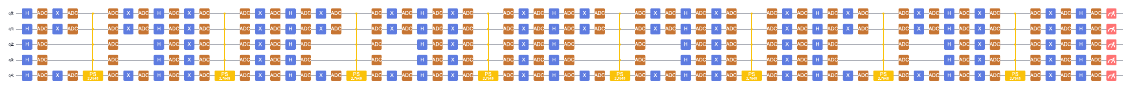}
		\setlength{\abovecaptionskip}{-0.01cm}
		\caption{The circuit  for the modified Grover’s algorithm targeting t = 01100 within a 5-qubit system.}
		\label{p17}
	\end{figure*}
	
	\begin{figure*}[htbp]
		\centering
		\includegraphics[width=2in]{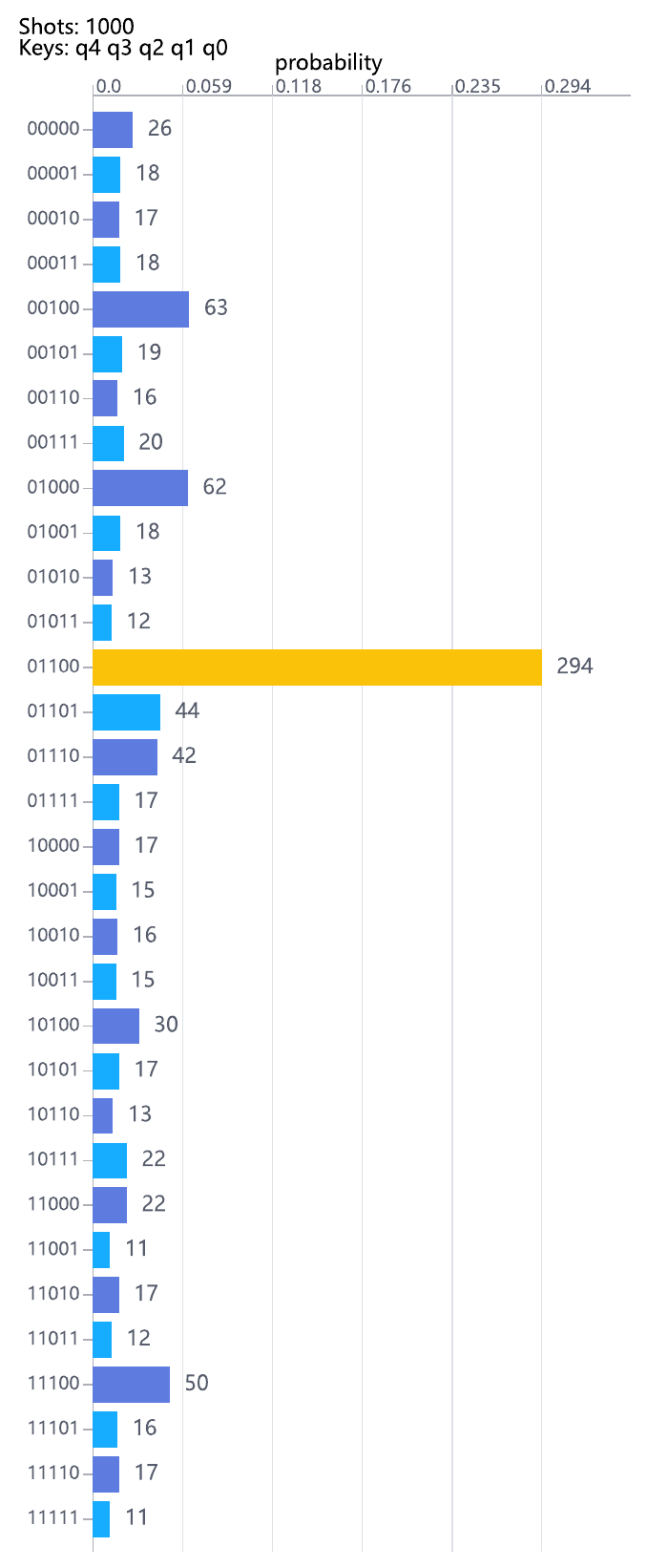}
		\setlength{\abovecaptionskip}{-0.01cm}
		\caption{Sampling results  of the circuit in {Fig. 19}.}
		\label{p18}
	\end{figure*}
	
	\begin{figure*}[htbp]
		\centering
		\includegraphics[width=4in]{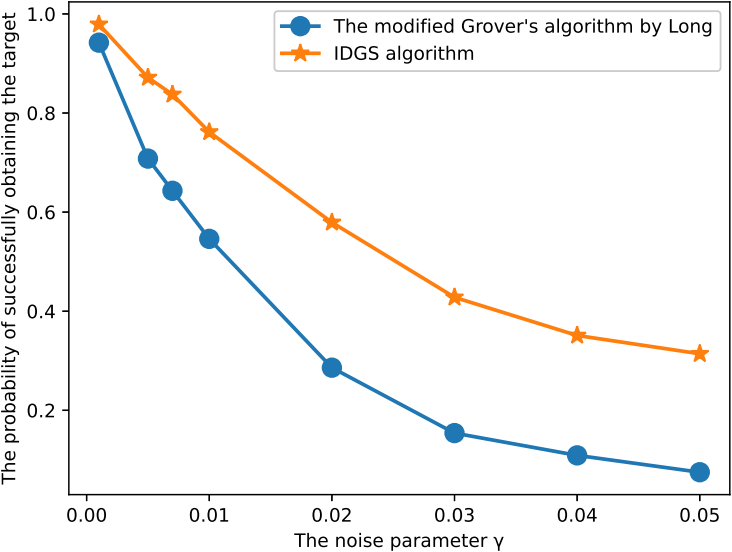}
		\setlength{\abovecaptionskip}{0.04cm}
		\caption{Success probability  of the algorithm under amplitude damping noise ($n=5$).}
		\label{p19}
	\end{figure*}

\subsection{12-qubit system  under the phase damping channel}
Based on the previous experiments on the 12-qubit system, we consider the robustness of the algorithm under the phase damping channel.

The phase damping channel describes the loss of quantum information without being accompanied by energy loss\cite{Nielsen2001}.  For any quantum state $\rho$,  if it undergoes the phase damping channel $\Phi$, then it {is transformed as follows:}
$$
\Phi (\rho )=\sum_{i=0}^{1} E_i \rho E_i^{\dagger}, 
$$
where
\[
E_0=\begin{pmatrix}
  1 & 0  \\
  0 &  \sqrt{1-\gamma}
\end{pmatrix}, E_1=\begin{pmatrix}
  0 &  0\\
  0 &  \sqrt{\gamma}
\end{pmatrix},
\]
 $E_i$ are Kraus operators satisfying $\sum_{i=0}^{1}E_i^\dag
E_i=I$ with $I$ being the identity operator, and  $0<\gamma<1$. 

We introduce noise after each quantum operation. We set the noise parameter $\gamma\in\{0.001, 0.002, 0.003, 0.005, 0.007\}$. Since the detailed circuit  for the 5-qubit system with added noise in the algorithm has been provided in the previous section, the situation for the larger-scale case is similar. However, the complete noise-inclusive circuit is relatively large, so we directly present the  results here. The success probability is defined as the ratio of the number of successful measurements to the total number of measurements.
  The success probability of our algorithm is based on Eq. (\ref{ss1}). Under the same noise parameter settings, we also tested the modified Grover’s algorithm by Long. The overall results are shown in Fig. \ref{pc}. As can be seen from the figure, our algorithm is more resistant to noise.
\begin{figure*}[htbp]
		\centering
		\includegraphics[width=4in]{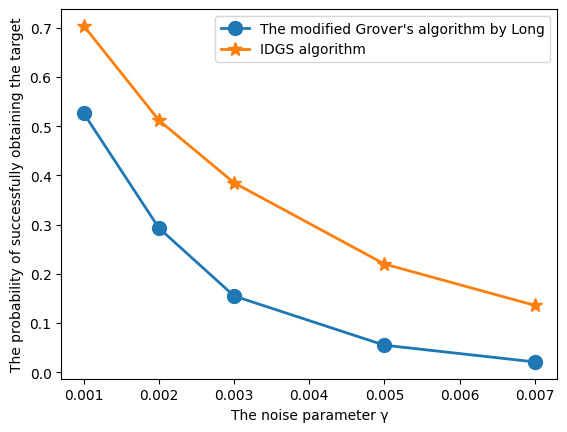}
		\setlength{\abovecaptionskip}{0.04cm}
		\caption{Success probability  of the algorithm under phase damping noise ($n=12$).}
		\label{pc}
	\end{figure*}

\section{Conclusion}\label{s6}
In this article, we {proposed} an exact quantum partial search algorithm. Then, we  presented an algorithm for  constructing a query operator from Boolean functions to their subfunctions. 
By dividing the target string of the search problem into several substrings and integrating the query operator of each subfunction,  we {introduced} a precise low-depth distributed quantum search algorithm, i.e., {the} IDGS algorithm.  {The} IDGS algorithm  searches for the first $p$ bits of the target each time, and then uses the found partial information  along with a modified version of Grover's algorithm to search for the remaining $n-p-k$ bits. When $p=n-k$, {the} IDGS algorithm is similar to the algorithm for single-objective cases in Ref. \cite{Qiu2022}, but our algorithm demonstrates enhancements in terms of accuracy.    Initially, we established the theoretical validity of our algorithm, {and then conducted}  simulations with the quantum software MindQuantum, by which we successfully identified the target information with  probability 1. 
In testing the robustness of {the} IDGS algorithm, we introduced amplitude damping  {and phase damping} noise into the experiments. It was observed that the IDGS algorithm maintained a certain probability of successfully locating the target, even in instances where  Long's modified version of Grover's algorithm failed to do so. This outcome demonstrates the noise-resistant capabilities of the IDGS algorithm.

Compared to exact search algorithms, {the} IDGS algorithm can be implemented on small-scale quantum computers, and on each computing node, our algorithm {has}  fewer input qubits and lower circuit depth.  Therefore, {the} IDGS algorithm {is} a more viable option for real-world applications where computational resources and specific hardware capabilities may be limited.
Further consideration can be given to how to design a  low-depth distributed precise search algorithm with parallel execution for {multi-objective} scenarios.

\section*{Acknowledgements}
The authors are thankful  to the anonymous referees and editor for their significant comments and suggestions that greatly help us improve the quality
of the manuscript.  This work is supported in part by the NSFC (No. 61876195).
\section*{Declaration of competing interest}
The authors declare that they have no known competing financial interests or personal relationships that could have appeared to influence the work reported in this paper.
\appendix
\section{Proof of Proposition 1}
\label{app1}
\begin {proof}
Let  $$|\psi_2\rangle= |t^*\rangle\otimes\left(g_{t}|t'\rangle
+ b_{t} \sum_{x\neq t'} \frac{|x\rangle}{\sqrt{2^{n-q}-1}}\right)
+   \sum_{x^{''}\neq t^*}|x^{''}\rangle\otimes \frac{\sqrt{2^{n-q}}\cos \left(\left(2 j_1+1\right) \lambda\right)}{\sqrt{2^{n}-1}}|\varphi_{{n-q}}\rangle.$$ 
{In    {Step} 3 of the algorithm, we have}
 \begin{align*}
 &G|\psi_2\rangle\\
 =&U_{n}U_{f}\left(|t^*\rangle\otimes\left(g_{t}|t'\rangle
+ b_{t} \sum_{x\neq t'} \frac{|x\rangle}{\sqrt{2^{n-q}-1}}\right)
+   \sum_{x^{''}\neq t^*}|x^{''}\rangle\otimes \frac{\sqrt{2^{n-q}}\cos \left(\left(2 j_1+1\right) \lambda\right)}{\sqrt{2^{n}-1}}|\varphi_{{n-q}}\rangle\right)\\
=&\left(2|\varphi_n\rangle\langle\varphi_n|-I_{n}\right)\left( -g_{t}|t^*t'\rangle
+  \sum_{x\neq t'} \frac{b_{t}|t^*x\rangle}{\sqrt{2^{n-q}-1}}
+   \sum_{x^{''}\neq t^*}|x^{''}\varphi_{{n-q}}\rangle \frac{\sqrt{2^{n-q}}\cos \left(\left(2 j_1+1\right) \lambda\right)}{\sqrt{2^{n}-1}}\right)\\
=&\frac{-2g_{t}}{\sqrt{2^{n}}}|\varphi_{{n}}\rangle+g_{t}|t^*t'\rangle
+\frac{2b_{t}(2^{n-q}-1)}{\sqrt{2^{n-q}-1}\sqrt{2^{n}}}|\varphi_{{n}}\rangle-b_{t} \sum_{x\neq t'} \frac{|t^*x\rangle}{\sqrt{2^{n-q}-1}}\\
+&\cos \left(\left(2 j_1+1\right) \lambda\right)\left[2\frac{ (2^{n}-2^{n-q})|\varphi_{{n}}\rangle}{\sqrt{2^{n}(2^{n}-1)}}- \sum_{x^{''}\neq t^*}\frac{\sqrt{2^{n-q}}|x^{''}\varphi_{n-q}\rangle}{\sqrt{2^{n}-1}}\right].\\
=&\left(g_{t}+\frac{2b_{t}\sqrt{2^{n-q}-1}-2g_{t}}{2^{n}}+2\cos \left(\left(2 j_1+1\right) \lambda\right)\frac{ 2^{n}-2^{n-q}}{2^{n}\sqrt{2^{n}-1}
}\right) |t^*t'\rangle\\
+&\left(-\frac{b_{t}}{\sqrt{2^{n-q}-1}}+\frac{2b_{t}\sqrt{2^{n-q}-1}-2g_{t}}{2^{n}}+2\cos \left(\left(2 j_1+1\right) \lambda\right)\frac{ 2^{n}-2^{n-q}}{2^{n}\sqrt{2^{n}-1}}\right)\sum_{x\neq t'} |t^*x\rangle\\
+&\left(\frac{\left(\sqrt{2^{n}-1}\right)\left(2b_{t}\sqrt{2^{n-q}-1}-2g_{t}\right)+\cos \left(\left(2 j_1+1\right) \lambda\right)\left(2^{n}-2^{n-q+1}\right)}
{2^{n}\sqrt{2^{n}-1}}\right)\sum_{x^{''}\neq t^*}\sqrt{2^{n-q}}|x^{''}\varphi_{n-q}\rangle.
\end{align*} 
Substituting Eq. (\ref{eqa cancellation different}), we obtain
 \begin{equation}\label{a10}
 \begin{aligned}
 G|\psi_2\rangle&=\left(g_{t}+\frac{2b_{t}\sqrt{2^{n-q}-1}-2g_{t}}{2^{n}}+2\cos \left(\left(2 j_1+1\right) \lambda\right)\frac{ 2^{n}-2^{n-q}}{2^{n}\sqrt{2^{n}-1}
}\right) |t^*t'\rangle\\
&+\left(-\frac{b_{t}}{\sqrt{2^{n-q}-1}}+\frac{2b_{t}\sqrt{2^{n-q}-1}-2g_{t}}{2^{n}}+2\cos \left(\left(2 j_1+1\right) \lambda\right)\frac{ 2^{n}-2^{n-q}}{2^{n}\sqrt{2^{n}-1}}\right) \sum_{x\neq t'} |t^*x\rangle.
 \end{aligned} 
\end{equation}
Thus, the proposition holds.
\end {proof}

\section{Proof of Proposition 2}\label{x2}
\begin {proof}
First, we prove that the GRK quantum partial search algorithm can return the answer with probability 1, i.e.,  $j_1$ and $j_2$ satisfy the condition of Proposition \ref{p1}. {For} a large database {with} $2^{n}\rightarrow\infty$ and finite $q$, the following {equations hold}:
\begin{equation*}
\lambda=\arcsin{{\frac {1}{\sqrt {2^n}}}}\approx\frac {1}{\sqrt {2^n}}, \lambda'=\arcsin{\frac {1} {\sqrt{2^{n-q}}}}\approx\frac {1} {\sqrt{2^{n-q}}}.
\end{equation*}
Next, the following {equations} can be further derived:
\begin{equation}\label{a11} 
 (2j_1+1)\lambda=\frac {\pi}{2}-\frac {2\alpha_{q}}{\sqrt {2^q}}+\frac {1}{\sqrt {2^n}}\approx\frac {\pi}{2}-\frac {2\alpha_{q}}{\sqrt {2^q}},\quad 2j_2\lambda'=2\beta_{q}.
\end{equation}
Substituting Eq.  (\ref{a11}) into Eq.  (\ref{eqa cancellation different}), we have
 \begin{equation}\label{a12}
\tan\left(\frac{2\alpha_{q}}{\sqrt{2^q}}\right)=\frac{2\sqrt{2^q}\sin2\beta_{q}}{2^q-4\sin^2\beta_{q}}.
\end{equation}
Substituting  \begin{equation*}
\alpha_{q}=\frac{\sqrt {2^q}}{2}\arctan\left(\frac{\sqrt{3\cdot{2^q}-4}}{{2^q}-2}\right),\quad\quad \beta_{q}=\arcsin\sqrt{\frac{{2^q}}{4({2^q}-1)}},
\end{equation*} into Eq.  (\ref{a12}),  {we find that} {Eq.  (\ref{eqa cancellation different})} holds true.  {Therefore,} the probability of obtaining the correct answer is 1.

The number of queries of the GRK quantum partial search algorithm is $S=j_1+j_2+1 \approx \frac \pi 4 {\sqrt{2^{n}}}+ (\beta_{q}-\alpha_{q}){\sqrt {2^{n-q}}}$. If $\beta_{q}-\alpha_{q}$ is minimized, then  $S$ is minimized.  
  
Since the number of queries at each step of the algorithm must be non-negative and the total number of queries is  $$\frac \pi 4 {\sqrt{2^{n}}}+ (\beta_{q}-\alpha_{q}){\sqrt {2^{n-q}}}\le\frac \pi 4 {\sqrt{2^{n}}},$$  we have $0\le\beta_{q}\le\alpha_{q}\le\frac  \pi 4 {\sqrt{2^{q}}}$, which implies $0\le\frac{2\alpha_{q}}{\sqrt{2^q}}\le\frac {\pi} {2}$. From  Eq. (\ref{a12}) above, we can obtain 
 \begin{equation}\label{a13}
\alpha_{q}=\frac{\sqrt{2^q}}{2} \arctan \left(\frac{2 \sqrt{2^q} \sin 2 \beta_{q}}{2^q-4 \sin ^2 \beta_{q}}\right). 
\end{equation} 
Therefore,  we only need to find the minimum value of the function \begin{equation}f(\beta_{q})=\label{a14}\beta_{q}-\frac{\sqrt{2^q}}{2} \arctan \left(\frac{2 \sqrt{2^q} \sin 2 \beta_{q}}{2^q-4 \sin ^2 \beta_{q}}\right).\end{equation} 
Differentiating \( f \), we obtain$$
f^{\prime}(\beta_{q})=\frac{16(2^{q}-1) \sin ^4 \beta_{q}- {2^{2q+2}} \sin ^2 \beta_{q}+{2^{2q}}}{16(2^{q}-1) \sin ^4 \beta_{q}- 2^{q+3} \sin ^2 \beta_{q}-{2^{2q}}}.
$$
 According to the value of $\beta_{q}$, we have $\sin\beta_{q}=\sqrt{\frac{{2^q}}{4({2^q}-1)}}$. Substituting this into the expression for $f^{\prime}(\beta_{q})$, we get  \( f^{\prime}(\beta_{q})= 0 .\)

Next, calculating the second derivative of \( f \) gives:
$$
f^{\prime \prime}(\beta_{q})=\frac{2^{q+2} \sin 2 \beta_{q}\left[4(2^{q}-1)(2^{q}-2) \cos ^2 2\beta_{q}+16(2^{q}-1) \cos 2 \beta_{q}+(2^{q}-2)^2(2^{q}+2)\right]}{\left[16(2^{q}-1) \sin ^4 \beta_{q}-2^{q+3}\sin ^2 \beta_{q}-2^{2q}\right]^2} .
$$
If \( q > 1 \), it can be seen from Eq. (\ref{628}) that  $\sin2\beta_{q}>0$ and $\cos2\beta_{q}>0$, {and} thus \( f^{\prime \prime}(\beta_{q}) > 0 \). Therefore, 
\( f \) has a minimum value under Eq. (\ref{628}). If \( q =1 \), then 
$f^{\prime \prime}(\beta_{1})=0$ and 
the function 
 $$
f^{\prime}(\beta_{1})=\frac{(4 \sin ^2 \beta_{1}- 2)^2}{(4 \sin ^2 \beta_{1}- 2)^2-8}.
$$ Because \(0 \le\sin ^2 \beta_{1}\le1 \), it follows that \((4 \sin ^2 \beta_{1}- 2)^2-8<0\).  Therefore, \(f(\beta_{1})\) decreases as \(\beta_{1}\) increases. The derivative of $\alpha_{1}$ can be obtained as follows:$$\alpha_{1}^{\prime}=\frac{2\sqrt{2}}{1+\left(\frac{2 \sqrt{2} \sin 2 \beta_{1}}{2-4 \sin ^2 \beta_{1}}\right)^2}.  $$Obviously, $\alpha_{1}^{\prime}>0$. Consequently, 
$\alpha_{1}$ increases as $\beta_{1}$ increases. Additionally, since $0\le\beta_{1}\le\alpha_{1}\le\frac  \pi 4 {\sqrt{2}}$, the maximum value of $\alpha_{1}$ is $\frac  \pi 4 {\sqrt{2}}$. At this point, $\beta_{1}$ is $\frac  \pi 4$. If $\beta_{1}$ increases, $\alpha_{1}$ will exceed $\frac  \pi 4 {\sqrt{2}}$. When $\beta_{1}=\frac{\pi} {4}$, $\alpha_{1}=\frac{\pi} { 2 \sqrt{2}}$, \(f(\beta_{1})\) has a minimum value. In summary, the proposition is proven.
\end {proof}

\section{Proof of Lemma 1}\label{x3}
\begin{proof}
To prove Eq. (\ref{sgzong14}), we  simplify it as follows: 
\begin{align}
\left(a_{t^{\prime}}\cos \phi +E'\right)\left(\cos \theta +1\right)& =  a_{t^{\prime}}\sin\phi\sin\theta, \label{Xgzong1}\\
2^{n}F'+\left(a_{t^{\prime}}\cos \phi+E'\right)\left(\cos \theta -1\right)& =a_{t^{\prime}}\sin\phi\sin\theta,\label{Xgzong2}
\end{align}
where both sides of the first equation in Eq. (\ref{sgzong14}) are multiplied by $\cos \frac{\theta }{2}$, {with} $\cos \frac{\theta }{2}\neq 0$. {We have} $$\cos\theta +1 =\frac{2^{2n-1}F'^2-2^{n+1}E'F'+2E'^2-2a_{t^\prime}^2}{E'^2-2^{n}E'F'-a_{t^\prime}^2} , ~~\cos \theta-1 =\frac{2^{2n-1}F'^2}{E'^2-2^{n}E'F'-a_{t^\prime}^2}.$$Thus, calculating the {left-hand} side of Eq. (\ref{Xgzong1}) yields $$2^{n}F'+\frac{2^{3n-2}F'^3}{E'^2-2^{n}E'F'-a_{t^\prime}^2},$$  and calculating the {left-hand} side of Eq. (\ref{Xgzong2}) yields $$2^{n}F'+\frac{2^{3n-2}F'^3}{E'^2-2^{n}E'F'-a_{t^\prime}^2}.$$ Therefore, it suffices to prove that \begin{equation}\label{Xgzong3}
a_{t^{\prime}}\sin\phi\sin\theta=2^{n}F'+\frac{2^{3n-2}F'^3}{E'^2-2^{n}E'F'-a_{t^\prime}^2}.
\end{equation}
Simplifying the {right-hand} side of Eq. (\ref{Xgzong3}) yields $$\frac{2^{n}F'\left(\left(E'-2^{n-1}F'\right)^2-a_{t^\prime}^2\right)}{E'^2-2^{n}E'F'-a_{t^\prime}^2}.$$ Since $F' > 0$ and the  inequality (\ref{gzong1}) holds, it follows that $$\frac{2^{n}F'\left(\left(E'-2^{n-1}F'\right)^2-a_{t^\prime}^2\right)}{E'^2-2^{n}E'F'-a_{t^\prime}^2}>0.$$ We know that $$\sin\theta = \pm\sqrt{1-\cos^2\theta}=\pm 2^{n}F'\frac{\sqrt{a_{t^\prime}^2-\left(E'-2^{n-1}F'\right)^2}}{a_{t^\prime}^2-E'^2+2^{n}E'F'},$$ $$ \sin\phi= \pm \sqrt{1-\cos^2\phi}=\pm \sqrt{\frac{a_{t^\prime}^2-\left(E'-2^{n-1}F'\right)^2}{a_{t^\prime}^2}}.$$ Therefore, if $a_{t^{'}}>0$, then choose $\phi$ and $\theta$ such that $\sin\phi\sin\theta>0;$ if $a_{t^{'}}<0$, then choose $\phi$ and $\theta$ such that $\sin\phi\sin\theta<0.$ Consequently, Eq. (\ref{Xgzong3})  holds.
\end{proof}

\section{Proof of Theorem 1}\label{x4}
\begin{proof}
{After} {Step} 3 of the algorithm, we have
\begin{equation}\label{zong121}
 \begin{aligned}
 G_g{G_1}^{ j_2^{'}}G^{j_1^{'}}|\varphi_n\rangle=&\frac{a_{t^{'}}(1-e^{\imath\theta })e^{\imath\phi }}{\sqrt{2^{n}}}|\varphi_{{n}}\rangle-e^{\imath\phi }a_{t^{'}}|t^*t'\rangle\\
&+\frac{a_{n_{t^{'}}}(2^{n-q}-1)(1-e^{\imath\theta })}{\sqrt{2^{n-q}-1}\sqrt{2^{n}}}|\varphi_{{n}}\rangle-\sum_{x\neq t'}\frac{a_{n_{t^{'}}} |t^*x\rangle}{\sqrt{2^{n-q}-1}}\\
&+\cos \left((2 j_1^{'}+1) \lambda\right)\left[\frac{ (2^{n}-2^{n-q})(1-e^{\imath\theta })|\varphi_{{n}}\rangle}{\sqrt{2^{n}(2^{n}-1)}}-\sum_{x^{''}\neq t^*}\frac{\sqrt{2^{n-q}}|x^{''}\varphi_{n-q}\rangle}{\sqrt{2^{n}-1}}\right].
\end{aligned} 
\end{equation}
Since $$|\varphi_{{n}}\rangle=\frac{ |t^*t'\rangle+\sum\limits_{x\neq t'} |t^*x\rangle+\sum\limits_{x^{''}\neq t^*}\sqrt{2^{n-q}}|x^{''}\varphi_{n-q}\rangle}{\sqrt{2^{n}}},$$ after the action of $G_g$, the amplitude of {each element} in {the} non-target blocks is 
\begin{align}\label{zong1311}&\frac{a_{t^{'}}(1-e^{\imath\theta })e^{\imath\phi }}{2^{n}}+\frac{a_{n_{t^{'}}}(2^{n-q}-1)(1-e^{\imath\theta })}{2^{n}\sqrt{2^{n-q}-1}}\\ \nonumber
+&\frac{\cos \left(\left(2 j_1^{'}+1\right) \theta_1\right)(2^{n}-2^{n-q})(1-e^{\imath\theta })}{2^{n}\sqrt{2^{n}-1}}-\frac{\cos \left(\left(2 j_1^{'}+1\right) \lambda\right)}{\sqrt{2^{n}-1}}.\end{align}
Combining $E'$ and $F'$ {in} Eq. (\ref{zong1311}) yields the following equation{:}
\begin{align}\label{1zong1511}
&\frac{(1-e^{\imath\theta })(a_{t^{'}}e^{\imath\phi }+E')}{2^{n}}-F' \nonumber \\
=&\frac{(1-\cos \theta-\imath\sin\theta  )[a_{t^{'}}(\cos \phi+\imath\sin\phi )+E']}{2^{n}}-F' \nonumber\\
=&\frac{(1-\cos \theta)(a_{t^{'}}\cos \phi+E' )+a_{t^{'}}\sin\theta\sin\phi }{2^{n}}-F'\nonumber\\
&+\imath\frac{[a_{t^{'}}(1-\cos \theta)\sin\phi -a_{t^{'}}\sin\theta \cos \phi-E'\sin\theta ]}{2^{n}}.
\end{align}
Finally, substituting Eq.  (\ref{sgzong14})  into Eq. (\ref{1zong1511}) yields  $$\frac{(1-e^{\imath\theta })(a_{t^{'}}e^{\imath\phi }+E')}{2^{n}}-F'=0,$$i.e.,  the non-target block amplitude is 0. Hence, the target  is finally obtained by measurement. 
\end{proof}

\section{Proof of Lemma 2}\label{x5}
\begin {proof}  
 In the basis $\{\sum\limits_{\substack{x=0
 \\x \neq x_{ \widetilde{i}}y_t}}^{2^{n-k}-1} \frac{|x\rangle}{\sqrt{2^{n-k}-1}},  |x_{ \widetilde{i}}y_t\rangle\}$,  the matrix form of $G_{2}$ is $$
\begin{pmatrix}
\cos 2\theta_{1} & -\sin 2\theta_{1}  \\
\sin 2\theta_{1} & \cos 2\theta_{1}
\end{pmatrix}
.$$
The eigenvalues of $G_{2}$ are $\lambda_{1}^{ \pm}=\exp \left[ \pm 2 \imath \theta_{1}\right]$, and the eigenvectors corresponding to the eigenvalues 
$\lambda_{1}^{ \pm}$ are $( \frac{\imath}{\sqrt{2}}, \frac{1}{\sqrt{2}})^T$ and $( \frac{-\imath}{\sqrt{2}}, \frac{1}{\sqrt{2}})^T$, where 
$T$ denotes the transpose of a vector. Expressing the eigenvectors in terms of the basis $\{\sum\limits_{\substack{x=0
 \\x \neq x_{ \widetilde{i}}y_t}}^{2^{n-k}-1} \frac{|x\rangle}{\sqrt{2^{n-k}-1}},  |x_{ \widetilde{i}}y_t\rangle\}$ gives $$\left|\psi^{ \pm}_{1}\right\rangle =\frac{1}{\sqrt{2}}|x_{ \widetilde{i}}y_t\rangle \pm \frac{\imath}{\sqrt{2}}\left(\sum\limits_{\substack{x=0
 \\x \neq x_{ \widetilde{i}}y_t}}^{2^{n-k}-1} \frac{|x\rangle}{\sqrt{2^{n-k}-1}}\right).$$
Thus, we have $$G_{2}\left|\psi^{ \pm}_{1}\right\rangle=\lambda_{1}^{ \pm}\left|\psi^{ \pm}_{1}\right\rangle. $$
Therefore, the state after   {Step} 3 is 
\begin{equation}  \label{G1}
G_{2}^{p_{1}}|\varphi_{{n-k}}\rangle=\sin \left(\left(2 p_{1}+1\right) \theta_{1}\right)|x_{ \widetilde{i}}y_t\rangle+ \sum_{\substack{x=0 \\
x \neq x_{ \widetilde{i}}y_t}}^{2^{n-k}-1}  \frac{\cos \left(\left(2 p_{1}+1\right) \theta_{1}\right)|x\rangle}{\sqrt{2^{n-k}-1}}.
\end{equation}  For $\{0,1\}^{n-k}$, based on whether the first $p$ bits are the same, we provide the following $2^{p}$ blocks:  $$ \overbrace{0 \cdots000}^{p}y_1,~ \overbrace{0 \cdots001}^{p}y_2, ~\cdots, ~ \overbrace{1 \cdots111}^{p}y_{2^{p}}, y_{m} \in \{0,1\}^{n-p-k}, m=1,\cdots,2^{p}.$$ 
According to this partitioning method, {by combining the terms on} the {right-hand} side of Eq. (\ref{G1}), we can express $|\psi_{\widetilde{i}_{n-k}}\rangle$ as Eq. (\ref{G2}).
\end{proof}

\section{Proof of Lemma 3}\label{x6}
\begin {proof} 
 In the basis $\{|x_{ \widetilde{i}}\rangle\otimes\sum\limits_{x^{\prime}\neq y_t} \frac{|x^{\prime}\rangle}{\sqrt{2^{n-p-k}-1}},  |x_{ \widetilde{i}}y_t\rangle\}$,  the matrix form of $G_{3}$ is $$
\begin{pmatrix}
\cos 2\theta' & -\sin 2\theta'  \\
\sin 2\theta' & \cos 2\theta'
\end{pmatrix}
.$$
 {Following the procedure}  in Lemma \ref{l2}, we   obtain the eigenvalues of $G_{3}$ and  {the} corresponding eigenvectors as follows: 
 $$G_3\left|\psi^{ \pm}_0\right\rangle=\lambda_2^{ \pm}\left|\psi^{ \pm}_0\right\rangle, \lambda_2^{ \pm}=\exp \left( \pm 2 \imath \theta'\right), $$
$$\left|\psi^{ \pm}_0\right\rangle =|x_{ \widetilde{i}}\rangle\otimes\left[\frac{1}{\sqrt{2}}|y_t\rangle \pm \frac{\imath}{\sqrt{2}}\left(\sum\limits_{x^{\prime}\neq y_t} \frac{|x^{\prime}\rangle}{\sqrt{2^{n-p-k}-1}}\right)\right].$$
From the eigenvectors of $G_3$, we can derive the following expression:$$|x_{ \widetilde{i}}y_t\rangle=\frac{\sqrt{2}}{2}\left(|\psi^{+}_0\rangle+|\psi^{-}_0\rangle\right), |x_{ \widetilde{i}}\rangle\otimes\sum\limits_{x^{\prime}\neq y_t} \frac{|x^{\prime}\rangle}{\sqrt{2^{n-p-k}-1}}=\frac{\imath\sqrt{2}}{2}\left(|\psi^{-}_0\rangle-|\psi^{+}_0\rangle\right).$$
  When $G_3$ acts on {the} target block, we have
\begin{align*}
 &G_{3}^{p_2}\left(|x_{ \widetilde{i}}\rangle\otimes\left(\sin \left(\left(2 p_1+1\right) \theta_1\right)|y_t\rangle+\sum_{x^{\prime}\neq y_t} \frac{\cos \left(\left(2 p_1+1\right) \theta_1\right)}{\sqrt{2^{n-k}-1}}|x^{\prime}\rangle\right)\right)\\
 =&G_{3}^{p_2}\left(\sin \left(\left(2 p_1+1\right) \theta_1\right)|x_{ \widetilde{i}}y_t\rangle+|x_{ \widetilde{i}}\rangle\otimes\sum\limits_{x^{\prime}\neq y_t} \frac{\cos \left(\left(2 p_1+1\right) \theta_1\right)|x^{\prime}\rangle}{\sqrt{2^{n-p-k}-1}}\right)\\
 =&G_{3}^{p_2}\left(\frac{\sin \left(\left(2 p_1+1\right) \theta_1\right)}{{\sqrt{2}}}\left(|\psi^{+}_0\rangle+|\psi^{-}_0\rangle\right)+ \frac{\imath{\sqrt{2^{n-p-k}-1}}\cos \left(\left(2 p_1+1\right) \theta_1\right)}{\sqrt{2}\sqrt{2^{n-k}-1}}\left(|\psi^{-}_0\rangle-|\psi^{+}_0\rangle\right)\right)\\
  =&\frac{\sin \left(\left(2 p_1+1\right) \theta_1\right)}{{\sqrt{2}}}\left(e^{2 \imath{p_2} \theta'}|\psi^{+}_0\rangle+e^{-2 \imath{p_2} \theta'}|\psi^{-}_0\rangle\right)\\
  +& \frac{\imath{\sqrt{2^{n-p-k}-1}}\cos \left(\left(2 p_1+1\right) \theta_1\right)}{\sqrt{2}\sqrt{2^{n-k}-1}}\left(e^{-2 \imath{p_2} \theta'}|\psi^{-}_0\rangle-e^{2 \imath{p_2} \theta'}|\psi^{+}_0\rangle\right)\\
  =&\left[\sin \left(\left(2 p_1+1\right) \theta_1\right)\cos \left(2 p_2 \theta'\right)+\frac{{\sqrt{2^{n-p-k}-1}}\cos \left(\left(2 p_1+1\right) \theta_1\right)\sin \left(2 p_2 \theta'\right)}{\sqrt{2^{n-k}-1}}\right]|x_{ \widetilde{i}}y_t\rangle\\
  +&\left[-\sin \left(\left(2 p_1+1\right) \theta_1\right)\sin \left(2 p_2 \theta'\right)+\frac{{\sqrt{2^{n-p-k}-1}}\cos \left(\left(2 p_1+1\right) \theta_1\right)\cos \left(2 p_2 \theta'\right)}{\sqrt{2^{n-k}-1}}\right]\sum\limits_{x^{\prime}\neq y_t} \frac{|x_{ \widetilde{i}}x^{\prime}\rangle}{\sqrt{2^{n-p-k}-1}}\\
  =&|x_{ \widetilde{i}}\rangle\otimes\left(a_{t}|y_t\rangle+a_{n_t}\sum_{x^{\prime}\neq y_t} \frac{|x^{\prime}\rangle}{\sqrt{2^{n-p-k}-1}}\right).
\end{align*}When $G_3$ acts on each non-target block, we have 
\begin{align*}
& G_{3} \left(|\psi_{{p}}\rangle\otimes \frac{\sqrt{2^{n-p-k}}\cos \left(\left(2 p_1+1\right) \theta_1\right)}{\sqrt{2^{n-k}-1}}|\varphi_{{n-p-k}}\rangle\right)\\
 =& \frac{|\psi_{{p}}\rangle\otimes\sqrt{2^{n-p-k}}\cos \left(\left(2 p_1+1\right) \theta_1\right)U_{n-p-k}|\varphi_{{n-p-k}}\rangle}{\sqrt{2^{n-k}-1}}\\
=&|\psi_{{p}}\rangle\otimes \frac{\sqrt{2^{n-p-k}}\cos \left(\left(2 p_1+1\right) \theta_1\right)}{\sqrt{2^{n-k}-1}}|\varphi_{{n-p-k}}\rangle.
\end{align*}
Thus, we have $G_{3}^{p_2}|\psi_{\widetilde{i}_{n-k}}\rangle=|\psi_{{\widetilde{i}}_{n-k}}^{\prime}\rangle.$

\end{proof}


\end{document}